\pdfoutput=1

\documentclass[a4paper,12pt,numbers=endperiod,headinclude,abstract=on]{scrartcl}

\usepackage[american]{babel}
\usepackage[T1]{fontenc}
\usepackage{lmodern}
\usepackage{microtype}
\usepackage{graphicx}
\usepackage{xcolor}
\usepackage{amsmath,amsfonts,amssymb,amsxtra}
\usepackage[thmmarks,amsmath]{ntheorem}
\usepackage[para]{footmisc}
\usepackage{setspace}
\usepackage{mathrsfs}
\usepackage{dsfont}
\usepackage{verbatim}
\usepackage{relsize}
\usepackage{float}
\usepackage{upgreek}
\usepackage{enumitem}
\usepackage{bm}
\usepackage{fixcmex}

\usepackage{geometry}
\usepackage{subcaption}
\newcounter{dummy}
\makeatletter
\newcommand\myitem[1][]{\item[#1]\refstepcounter{dummy}\def\@currentlabel{#1}}
\makeatother

\usepackage{bookmark}
\hypersetup{
	colorlinks,
	linkcolor={red!50!black},
	citecolor={blue!50!black},
	urlcolor={blue!50!black}
}

\theoremstyle{break}
\newtheorem{theorem}{Theorem}[section]
\newtheorem{proposition}[theorem]{Proposition}
\newtheorem{corollary}[theorem]{Corollary}
\newtheorem{lemma}[theorem]{Lemma}

\theorembodyfont{\normalfont}
\newtheorem{definition}[theorem]{Definition}

\theoremstyle{nonumberbreak}
\theoremheaderfont{\bfseries}
\theorembodyfont{\normalfont}
\theoremseparator{.}
\theoremsymbol{\linebreak[1]\hspace*{.5em plus 1fill}\ensuremath{\square}}
\newtheorem{proof}{Proof}

\numberwithin{equation}{section}

\allowdisplaybreaks

\DeclareMathOperator{\sgn}{sgn}
\DeclareMathOperator{\diag}{diag}

\newcommand{\Cinf}{C_{\mathrm{c}}^{\infty}}
\newcommand{\Hcal}{\mathcal{H}}
\newcommand{\Fcal}{\mathcal{F}}
\newcommand{\Acal}{\mathcal{A}}
\newcommand{\Rcal}{\mathcal{R}}
\newcommand{\Tcal}{\mathcal{T}}
\newcommand{\R}{\mathds{R}}
\newcommand{\C}{\mathds{C}}
\newcommand{\M}{\mathds{M}}
\newcommand{\WR}{\mathds{W}}
\newcommand{\e}{\mathrm{e}}
\newcommand{\xsf}{\mathsf{x}}
\newcommand{\esf}{\mathsf{e}}
\newcommand{\restr}{\!\!\upharpoonright}
\newcommand{\xvec}{\underline{x}}
\newcommand{\yvec}{\underline{y}}
\newcommand{\pvec}{\underline{p}}
\newcommand{\vel}{\mathtt{v}}
\newcommand{\dop}{\mathtt{d}}
\newcommand{\gam}{\upgamma}
\newcommand{\Teff}{\widehat{T}}
\renewcommand*{\vec}[1]{\underline{#1}}
\newcommand*\diff{\mathop{}\!\mathrm{d}}
\newcommand*\Diff[1]{\mathop{}\!\mathrm{d}^#1}
\newcommand{\normord}[1]{\, :\mathrel{\mkern2mu #1 \mkern2mu}:}
\renewcommand*{\Re}{\operatorname{Re}}

\newcommand{\mysepline}{\noindent\makebox[\linewidth]{\resizebox{0.3333\linewidth}{1pt}{$\bullet$}}\bigskip}

\usepackage{scalerel}
\newcommand{\nullmatrix}{\hstretch{0.8}{\mathds{O}}}

\usepackage{authblk}

\renewcommand{\thefootnote}{\fnsymbol{footnote}}

\title{{\Large Probing non-equilibrium steady states of the\\ Klein-Gordon field with Unruh-DeWitt detectors}}

\author[1,2]{Albert Georg Passegger\thanks{e-mail: \texttt{albert\_georg.passegger@uni-leipzig.de}}}
\author[1,3]{Rainer Verch\thanks{e-mail: \texttt{rainer.verch@uni-leipzig.de}}}
\affil[1]{Institut f\"{u}r Theoretische Physik, Universit\"{a}t Leipzig, Leipzig, Germany}
\affil[2]{Max Planck Institute for Mathematics in the Sciences, Leipzig, Germany}
\affil[3]{CY Advanced Studies, CY Cergy Paris Universit\'{e}, Neuville-sur-Oise, France}

\date{\vspace{-5mm}}

\begin{document}

\pagenumbering{arabic}

\maketitle

\vspace{-9mm}
\renewcommand{\thefootnote}{\arabic{footnote}}

\begin{abstract}
	\noindent We calculate the transition rate of an Unruh-DeWitt detector coupled to a non-equilibrium steady state (NESS) of a free massless scalar field on four-dimensional Minkowski spacetime. 
	Bringing two semi-infinite heat baths at different temperatures into thermal contact along a surface, the NESS arises at asymptotically late times as a stationary state that has modewise thermal properties and features a heat flow between the reservoirs. 
	The detector couples linearly to the field by a monopole interaction, and it moves inertially along the axis of the NESS heat flow. 
	We contrast the transition rate with the case of a detector that is coupled to an inertial thermal equilibrium state. 
	The results illustrate that the monopole does not couple to the heat flow, causing the detector to only register kinematical effects. 
	Hence dynamical features of the NESS are hidden from this detector model.
\end{abstract}

\section{Introduction}
\label{sec:intro}

Non-equilibrium steady states (NESS) \cite{Ruelle-NESS,Jaksic-Pillet-NESS-math} describe configurations of systems that are out of thermal equilibrium, but still are stationary under the time evolution. 
NESS of quantum systems have been considered in various models and contexts, and we refer to \cite{Ruelle-entropy-cmp,Ruelle-entropy,Jaksic-Pillet-NESS,Ogata2004,Merkli-Mueck-Sigal,Merkli-Mueck-Sigal-NESS,AbouSalem2007} for a limited selection of works. 
A conceptually simple class of NESS may arise from bringing two semi-infinite heat baths (thermal reservoirs) at different temperatures into thermal contact, a model that is also referred to as a ``quench'' \cite{DLSB} (see also \cite{Bernard-Doyon2012,Bernard-Doyon2015,Hoogeveen-Doyon2015,BDLS2015,Hollands-Longo2018} in the context of conformal field theory, \cite{Spohn-Lebowitz1977,Tasaki2001,Aschbacher-Pillet2003} for similar models, and further literature referenced in these works). 
Under the energy exchange between the baths, such a system can settle at asymptotically late times in a stationary configuration with a heat (energy) flow from the higher to the lower temperature heat bath. 

A NESS of this kind has been constructed in \cite{DLSB,Hack-Verch} for initial heat baths described by thermal equilibrium (KMS) states of scalar (Klein-Gordon) quantum fields on Minkowski spacetime coupled along a spatial hypersurface. 
The NESS, which exhibits a steady heat flow between the semi-infinite heat baths, is invariant under translations, and thermal at the temperature of the left and right bath for modes that are right- and left-moving along the longitudinal axis, respectively.\medskip 

In the present paper we probe the NESS of \cite{DLSB,Hack-Verch} for a free massless scalar field on four-dimensional Minkowski spacetime using a two-level Unruh-DeWitt monopole detector \cite{Unruh1976,DeWitt,Birrell-Davies,Unruh-Wald1984,Takagi1986} that moves with constant velocity along the longitudinal axis of the setup. 
The Unruh-DeWitt detector is an important ``particle detector'' model that is frequently used as a probe for quantum fields in studies on, for instance, the Unruh effect and its generalization along different stationary trajectories (see, e.g., \cite{Crispino-Higuchi-Matsas,Good-JA-Moustos-Temirkhan2020,Biermann-et-al2020} and references therein). 
We derive the asymptotic transition rate (number of transitions per unit proper time) of the detector in first-order time-dependent perturbation theory under the interaction with the NESS and contrast it with the transition rates of inertially moving detectors coupled to inertial thermal equilibrium states characterized by the KMS (Kubo-Martin-Schwinger) condition \cite{Kubo1957,Martin-Schwinger1959,HHW1967,Haag1996,Bratteli-Robinson2}.\medskip 

The topic is motivated by our recent work \cite{PV-disj}, where we discussed the late time behavior of a coupled system involving an Unruh-DeWitt detector that moves with constant velocity relative to an inertial KMS state of a massless scalar field. 
The results of \cite{PV-disj} indicate that the non-thermalization of the system relative to the inertial rest frame of the detector can be traced back to the disjointness of KMS states in different inertial frames (i.e.\ two distinct KMS states in different inertial frames cannot be generated from one another via quasi-local operations). 
That inertially moving detectors in a heat bath will not approach thermal equilibrium is supported by the transition rate calculated in \cite{Costa-Matsas1995}, which shows that the detector is affected by the motion-induced Doppler shift of the bath's thermal radiation relative to the detector's rest frame (see also \cite{Costa-Matsas-background1995,Landsberg-Matsas1996,Landsberg-Matsas2004}). 
The form of the asymptotic density matrix of the detector calculated in \cite{Papadatos-Anastopoulos2020}, and the impossibility for a state to satisfy the KMS condition relative to two different inertial frames \cite{Sewell2008,Sewell-rep2009}, further corroborate this picture. 

Based on these observations one might think that the anisotropy and apparent heat flow of a thermal reservoir from the perspective of the moving detector's rest frame is modeled by a NESS. 
However, it is important to note that the NESS we consider here does not satisfy the KMS condition with respect to any inertial reference frame (except for the free massless scalar field in two spacetime dimensions \cite{DLSB}; see also \cite[Sec.\ 3.2]{Hack-Verch}). 
Still, depending on the model, the response of a detector in terms of its transition rate could be comparable for a moving detector in a heat bath and in a NESS. 
We thus opt for an investigation in a detector-dependent setup, focusing on the question of what an Unruh-DeWitt monopole detector can ``see'' from a NESS in its transition rate.\medskip 

In Section \ref{sec:setup} the basic setup is introduced. 
We review the transition rate of an Unruh-DeWitt detector coupled to a massless Klein-Gordon field and briefly introduce the NESS of the field according to \cite{Hack-Verch}. 
In Section \ref{sec:transition-rates} we calculate and discuss the transition rates of a detector at rest and moving inertially along the longitudinal axis of the NESS as well as relative to a KMS heat bath under a monopole interaction. 
Some aspects of the detector's effective detailed balance temperatures in these cases are discussed in Section \ref{sec:detailed-balance}, and we define the comoving frame of the NESS in Section \ref{sec:case-NESS-comoving}. 
Our work represents a first step in the study of NESS from the perspective of detector systems, a topic which, to the best of our knowledge, has not received much attention yet. 
Section \ref{sec:conclusions-outlook} summarizes our findings and outlines open problems that could be picked up in the future. 
Additional material (mostly consisting of previously established results) is collected in four appendices, one of them (Appendix \ref{appendix:cm}) providing a detailed derivation of the transition rate of a detector moving with constant velocity through a massless scalar field heat bath, which has been presented in \cite{Costa-Matsas-background1995,Costa-Matsas1995}.

\paragraph{Conventions and remarks} We use the ``mostly minus'' metric signature $(+,-,-,-)$, and physical units where the vacuum speed of light, reduced Planck constant, and Boltzmann constant are set to $1$. 
The plots were generated using the \LaTeX\ package PGFPlots (\url{https://ctan.org/pkg/pgfplots}). 
The surface (mesh) plots are colored (with interpolation shading) as a visual aid according to the value of the function ($z$-axis; ``lower = blue'', ``higher = yellow'').

\section{Preliminaries and setup}
\label{sec:setup}

\subsection{Detector and quantum field}
\label{sec:setup-field-detector}

A two-level Unruh-DeWitt detector \cite{Unruh1976,DeWitt,Birrell-Davies,Unruh-Wald1984,Takagi1986}, henceforth simply called a \emph{detector}, is a quantum system with Hilbert space $\C^2$ and Hamiltonian $H_D = \diag(E,0)$ with respect to orthonormal basis vectors $|E\rangle$ and $|0\rangle$ of eigenvalues $E\in\R\setminus\{0\}$ and $0$, respectively, where $E$ represents the energy gap of the system. 
Depending on the sign of $E$ the vector $|0\rangle$ represents the ground (for $E>0$) or excited state (for $E<0$), and vice versa for $|E\rangle$. 
A transition from $|0\rangle$ to $|E\rangle$ is thus interpreted either as an excitation ($E>0$) or a decay (de-excitation, $E<0$) of the detector. \medskip

The \emph{quantum field} under consideration is a free massless real scalar (Klein-Gordon) field on four-dimensional Minkowski spacetime $\M$, and we essentially follow the setup described in \cite{Hack-Verch}, which is based on the algebraic approach to quantum field theory \cite{Haag1996,Fewster-Rejzner-AQFT}. 
Points in Minkowski spacetime $\M\cong\R^4$ relative to coordinate axes of an inertial reference frame will be denoted by $x \equiv (x^0 , \xvec) \equiv (x^0 , x^1 , x^2 , x^3)$. 
The field equation is given by $(\partial_0^2 - \sum_{j=1}^3 \partial_j^2)\phi=0$, where $\partial_\mu := \frac{\partial}{\partial x^\mu}$ for $\mu\in\{0,1,2,3\}$. 
The quantum field theory can be described by a unital $^\ast$-algebra $\Acal$ of observables, generated by formal symbols $\{\mathds{1},\phi(f) : f\in\Cinf(\M)\}$ subject to the usual relations implementing linearity, hermiticity, canonical commutation relations, and the field equation; see, e.g., \cite{Khavkine-Moretti,Fewster-Rejzner-AQFT} for details. 
The generators $\phi(f)$ are interpreted as ``smeared field operators'', which motivates the common integral kernel notation $\phi(f)=\int_{\M} \phi(x)f(x)\diff x$. 
On $\Acal$ the \emph{time evolution} along the time direction determined by the $x^0$-axis is given by the one-parameter $^\ast$-automorphism group $\alpha=\{\alpha_t\}_{t\in\R}$ with
\begin{gather}
	\label{eq:alpha}
	\alpha_t (\phi(x^0,\xvec)) = \phi(x^0 + t,\xvec) \, .
\end{gather}
Smeared with a compactly supported smooth function $f\in\Cinf(\M)$ this means that $\alpha_t (\phi(f))=\phi(f\circ T_{-t})$ for the time translation isometries $T_t : (x^0,\xvec)\mapsto(x^0 + t,\xvec)$ on $\M$.\medskip

Let $\omega : \Acal\to\C$ be a state, i.e.\ a linear functional on $\Acal$ that is positive ($\omega(A^\ast A) \geq 0$ for all $A\in\Acal$), normalized ($\omega(\mathds{1})=1$), and continuous in the topology induced by the test function topology. 
There exists a $^\ast$-representation $\pi_\omega$ of $\Acal$ by linear operators on a dense subspace $\mathcal{D}_\omega$ of a Hilbert space $\Hcal_\omega$ and a unit vector $\Omega_\omega \in \mathcal{D}_\omega$, such that $\omega(A)=\langle\Omega_\omega , \pi_\omega (A)\Omega_\omega\rangle$ for all $A\in\Acal$ and $\pi_\omega (\Acal)\Omega_\omega = \mathcal{D}_\omega$. 
The data $(\pi_\omega , \Hcal_\omega , \mathcal{D}_\omega , \Omega_\omega)$ is called the GNS representation of $\omega$ and is unique up to unitary equivalence (see, e.g., \cite{Khavkine-Moretti,Fewster-Rejzner-AQFT}). 
The quantum field is represented by an operator-valued distribution $\Cinf(\M)\ni f\mapsto\Phi(f)$ with operators $\Phi(f):=\pi_\omega (\phi(f))=\int_{\M} \Phi(x)f(x)\diff x$ on $\mathcal{D}_\omega$, which is used to formally define $\Phi(x)$ for $x\in\M$. 
For instance, in the GNS representation of the Minkowski vacuum, which is the ground state with respect to every inertial time translation, we have $\Phi(x)=(2\pi)^{-3/2} \int_{\R^3} (2\lVert\pvec\rVert)^{-1/2} (\e^{ip\cdot x} a^\ast (\pvec) + \e^{-ip\cdot x} a (\pvec)) \Diff3 \pvec$ for $p\cdot x=\lVert\pvec\rVert x^0 - \pvec\cdot\vec{x}$ and annihilation and creation operators $a,a^\ast$ on the bosonic Fock space over the one-particle space $L^2 (\R^3)$, obeying the commutation relations $[a (\pvec),a^\ast (\pvec')]=\delta(\pvec-\pvec')$. 
For details on these topics we refer to \cite{Khavkine-Moretti,Fewster-Rejzner-AQFT}.\medskip 

An important class of states is given by \emph{quasi-free Hadamard states}. 
Quasi-free (also called ``Gaussian'') states $\omega$ are completely determined by their distributional two-point (positive-frequency Wightman) correlation function $W_\omega$,
\begin{gather}
	\label{eq:two-point}
	W_\omega (x,y) := \omega(\phi(x)\phi(y)) \equiv \langle\Omega_\omega , \Phi(x)\Phi(y) \Omega_\omega \rangle \, , \quad x,y\in\M \, ,
\end{gather}
in that their even $n$-point functions are sums of products of two-point functions, and their odd $n$-point functions vanish \cite{Kay-Wald,Khavkine-Moretti,Fewster-Rejzner-AQFT}. 
Hadamard states are states for which $W_\omega$ is of Hadamard form, meaning that it has a specific universal (state-independent) singular structure representing a short-distance behavior that approximates the vacuum state \cite{Kay-Wald,Radzikowski1996,Wald-book,Khavkine-Moretti}. 
For our purposes it suffices to characterize a Hadamard state $\omega$ by the property that 
\begin{gather*}
	\widetilde{W}_\omega := W_\omega-W_{\mathrm{vac}}
\end{gather*}
is a smooth function on $\M \times \M$ (see \cite{Khavkine-Moretti} and references therein), where $W_{\mathrm{vac}}$ is the two-point function of the Minkowski vacuum state given by
\begin{gather}
	\label{eq:vac-two-point}
	W_{\mathrm{vac}}(x,y) = \frac{1}{(2\pi)^3} \int\limits_{\R^3} \frac{1}{2\lVert\pvec\rVert} \e^{-i\lVert\pvec\rVert(x^0 - y^0)} \e^{i\pvec\cdot(\xvec-\yvec)} \Diff3\pvec
\end{gather}
in the sense of bidistributions (see, e.g., \cite{Fewster-Rejzner-AQFT,Khavkine-Moretti}). 
The smooth part of the two-point function of a Hadamard state will always be denoted with a tilde, as defined here. 

We will study quasi-free Hadamard states $\omega$ that are $\alpha$-invariant, which means that they are stationary with respect to the inertial time evolution in the sense that $\omega\circ\alpha_t = \omega$ for all $t\in\R$. 
In that case the time evolution is generated by a Hamiltonian $H_\omega$ on the GNS space $\Hcal_\omega$, that is, $\frac{\mathrm{d}}{\diff t} \pi_\omega (\alpha_t (A)) \restr_{t=0} \, = i[H_\omega,\pi_\omega (A)]$ for $A\in\Acal$. 
The domain of $H_\omega$ contains the dense subspace $\pi_\omega (\Acal)\Omega_\omega$, which is invariant under $H_\omega$. 
The $\alpha$-invariant quasi-free Hadamard states considered in this work are the NESS $\omega_N$ that will be introduced in Section \ref{sec:setup-ness} \cite{Hack-Verch}, and the \emph{KMS states} $\omega_\beta$ on $\Acal$ for KMS parameter $\beta>0$ with respect to $\alpha$ (which satisfy the Hadamard condition by \cite{Sahlmann-Verch}), representing inertial thermal equilibrium states of the field at inverse temperature $\beta$ \cite{HHW1967,Haag1996,Bratteli-Robinson2} (see also \cite[Sec.\ 2.1]{Hack-Verch}). 
The two-point function $W_\beta := W_{\omega_\beta}$ of $\omega_\beta$ is given by
\begin{gather}
	W_\beta (x,y) = \frac{1}{(2\pi)^3} \int\limits_{\R^3} \frac{1}{2\lVert\pvec\rVert} \left( \frac{\e^{i\lVert\pvec\rVert(x^0 - y^0)} \e^{-i\pvec\cdot(\xvec-\yvec)}}{\e^{\beta\lVert\pvec\rVert} - 1} - \frac{\e^{-i\lVert\pvec\rVert(x^0 - y^0)} \e^{i\pvec\cdot(\xvec-\yvec)}}{\e^{-\beta\lVert\pvec\rVert} - 1} \right) \, \Diff3\pvec
	\label{eq:thermal-two-point}
\end{gather}
in the sense of bidistributions (see, e.g., \cite{Fewster-Rejzner-AQFT}, and Appendix \ref{appendix:cm} for an equivalent representation). 
The vacuum two-point function in Eq.\ \eqref{eq:vac-two-point} is recovered in the limit $\beta\to\infty$.

\subsection{Coupling and transition rate}
\label{sec:setup-transition-rate}

Let $\omega$ be a quasi-free Hadamard state on the $^\ast$-algebra $\Acal$ of the field, which is represented as operators $\Phi(x)$ on a dense subspace of the GNS Hilbert space $\Hcal_\omega$ (see Section \ref{sec:setup-field-detector}), and $\xsf:\R\to\M$, $\tau\mapsto \xsf(\tau)$ a smooth timelike curve describing the worldline of the detector parametrized in proper time $\tau\in\R$. 
The free time evolution of the combined detector-field system is assumed to be generated by a Hamiltonian $H_D \otimes \mathds{1} + \mathds{1} \otimes H_\omega$ on $\C^2 \otimes \Hcal_\omega$ with respect to proper time. 
We let the detector-field system evolve under a (point-like) \emph{monopole interaction} linear in the field, implemented in the interaction picture by adding a Hamiltonian of the form (see, e.g., \cite{DeWitt,Birrell-Davies,Louko-Satz,Waiting-for-Unruh}) 
\begin{gather*}
	\lambda \chi(\tau) \mu(\tau) \otimes \Phi(\xsf(\tau)) \, .
\end{gather*}
Here, $\lambda\in\R$ is the coupling parameter, and $\chi$ is a real-valued, non-negative, smooth function, called the switching function, whose support determines the times during which the coupling between detector and field is active. 
The operator $\mu(\tau)=\e^{iH_D \tau} \mu(0) \e^{-iH_D \tau}$ is the monopole moment operator of the detector at time $\tau$ for some self-adjoint bounded operator $\mu(0)$ on $\C^2$ such that $|\langle E|\mu(0)|0\rangle|\neq 0$.\medskip 

Assume that initially the detector is prepared in the state represented by $|0\rangle$, and the field in the state $\omega$. 
The probability for the detector to undergo a transition to the state represented by $|E\rangle$ (with the field taking any state) under the interaction is $\lambda^2 |\langle E|\mu(0)|0\rangle|^2 \, \Fcal_{\omega,\chi}(E)$ in first-order perturbation theory in $\lambda$ \cite{Unruh1976,DeWitt,Birrell-Davies,Louko-Satz,Waiting-for-Unruh}, where $\Fcal_{\omega,\chi}$ is the \emph{response function} defined by (for suitable $\chi$, e.g.\ smooth compactly supported or Schwartz function) 
\begin{gather*}
	\Fcal_{\omega,\chi} (E) := \int\limits_{\R} \int\limits_{\R} \e^{-iE(\tau-\tau')} \chi(\tau) \chi(\tau') W_\omega (\xsf(\tau),\xsf(\tau')) \, \diff \tau \diff \tau'
\end{gather*}
for the two-point function $W_\omega$ of the state $\omega$ (Eq.\ \eqref{eq:two-point}). 
Apart from prescribed internal quantities of the detector system, the transition probability is therefore determined by the response function. 
Since $W_\omega (x,y)=\overline{W_\omega (y,x)}$, the response function is real-valued. 
Due to the bidistributional nature of the two-point function, closed form expressions of $W_\omega (x,y)$ usually need to be regularized by an $i\epsilon$-prescription. 
The Hadamard property of $\omega$ entails a good control over the singular part of $W_\omega$ \cite{Kay-Wald,Radzikowski1996,Khavkine-Moretti}. 
As discussed in \cite{Louko-Satz} one can specify functions $W_\omega^{(\epsilon)}$, $\epsilon>0$, with the property that smearing them with test functions and taking $\epsilon\to 0_+$ results in $W_\omega$; accordingly, the integration in the definition of the response function is understood as $\lim_{\epsilon\to 0_+} \int_{\R} \int_{\R} \e^{-iE(\tau-\tau')} \chi(\tau) \chi(\tau') W_\omega^{(\epsilon)} (\xsf(\tau),\xsf(\tau')) \, \diff \tau \diff \tau'$.\medskip 

For the worldlines $\tau\mapsto\xsf(\tau)$ and states $\omega$ we consider in this work, the pullback of the two-point function along $\xsf$ that appears in the response function only depends on the proper time difference between two points on the worldline, i.e.\ $W_\omega (\xsf(\tau),\xsf(\tau')) = W_\omega(\xsf(\tau-\tau'),\xsf(0))$ for $\tau,\tau'\in\R$. 
Let us define the distribution $w_\omega$ by
\begin{gather*}
	w_\omega (s):= W_\omega(\xsf(s),\xsf(0)) \quad \text{for}\ s\in\R \, .
\end{gather*}
In the limit of infinitely long interaction time, the number of detector transitions per unit proper time, the so-called transition rate, is given by the (distributional) Fourier transform of $w_\omega$ (see, e.g., \cite{Birrell-Davies,Schlicht2004,Louko-Satz,Waiting-for-Unruh,Bunney-Parry-Perche-Louko2024}). 
Omitting the technical details one can see this as follows (see \cite{Waiting-for-Unruh}, and the appendix of \cite{Bunney-Parry-Perche-Louko2024} for a general derivation that also discusses the various conventions appearing in the literature): 
Take a switching function of the form $\chi_\theta (\tau) = \chi(\tau/\theta)$ for some non-negative smooth function $\chi$ with $\chi(0)=1$ and characteristic time scale given by the parameter $\theta>0$ (this is referred to as ``adiabatic scaling'' in \cite{Waiting-for-Unruh}). 
For example, following \cite{Bunney-Parry-Perche-Louko2024}, one can consider a Gaussian profile $\chi_\theta (\tau)=\e^{-\frac{\pi\tau^2}{2\theta^2}}$, which converges pointwise to the constant $1$ as $\theta\to\infty$ (eternal constant coupling, no switching). 
The corresponding response function then becomes (after setting $\zeta:=\tau+\tau'$, $s:=\tau-\tau'$)
\begin{gather*}
	\Fcal_{\omega,\chi_\theta}(E) = \frac{1}{2} \int\limits_{\R} \int\limits_{\R} \e^{-\frac{\pi(\zeta^2 + s^2)}{4\theta^2}} \e^{-iEs} \, w_\omega (s) \, \diff\zeta \diff s = \theta \int\limits_{\R} \e^{-iEs} \, \e^{-\frac{\pi s^2}{4\theta^2}} \, w_\omega (s) \diff s \, .
\end{gather*}
The limit
\begin{gather}
	\label{eq:transition-rate}
	\Rcal_\omega (E) := \lim\limits_{\theta\to\infty} \frac{1}{\theta} \Fcal_{\omega,\chi_\theta}(E) = \int\limits_{\R} \e^{-iEs} w_\omega (s) \diff s
\end{gather}
represents the \emph{transition rate} of the detector (see \cite{Birrell-Davies,Schlicht2004,Louko-Satz,Waiting-for-Unruh,Bunney-Parry-Perche-Louko2024}). 
We will exclusively focus on transition rates of detectors in the stated sense, i.e., for infinitely long interaction and sufficiently small (constant) coupling parameter conforming to first-order perturbation theory. 
A formula for the (in general time-dependent) transition rate at some finite time during the interaction, which takes properties of the switching function into account and applies to possibly non-stationary detector worldlines, has been given in \cite{Schlicht2004} and addressed in \cite{Louko-Satz} for theories on curved spacetimes (see also \cite{Louko-Satz-click,Satz2007} for related discussions).

\subsection{NESS}
\label{sec:setup-ness}

Non-equilibrium steady states (NESS) \cite{Ruelle-NESS,Jaksic-Pillet-NESS-math} for linear and interacting Klein-Gordon models have been constructed in \cite{Hack-Verch}, generalizing previous results due to \cite{DLSB}. 
They are obtained as late time limits of two semi-infinite inertial heat baths at different inverse temperatures $\beta_L , \beta_R > 0$ that are brought into thermal contact, as will be outlined below. 
From \cite[Thm.\ 3.2 (1)]{Hack-Verch} (see also \cite{DLSB}) we gather the following definition of the NESS considered in the setup of the present paper. 
\begin{definition}[NESS of massless Klein-Gordon field]
	\label{def:ness}
	Let $\beta_L , \beta_R > 0$. 
	The NESS $\omega_N$ on the $^\ast$-algebra $\Acal$ of the free massless scalar field on $\M\cong\R^4$ is a quasi-free state (depending on $\beta_L , \beta_R$) with two-point function $W_N := W_{\omega_N}$ (Eq.\ \eqref{eq:two-point}) given by
	\begin{gather}
		\label{eq:ness-two-point}
		W_N (x,y) = \frac{1}{(2\pi)^3} \int\limits_{\R^3} \frac{1}{2\lVert\pvec\rVert} \e^{i\pvec\cdot(\xvec-\yvec)} \left( \frac{\e^{i\lVert\pvec\rVert(x^0 - y^0)}}{\e^{\upbeta(-p_1)\lVert\pvec\rVert} - 1} - \frac{\e^{-i\lVert\pvec\rVert(x^0 - y^0)}}{\e^{-\upbeta(p_1)\lVert\pvec\rVert} - 1} \right) \, \Diff3\pvec
	\end{gather}
	for $\pvec=(p_1,p_2,p_3)\in\R^3$, where 
	\begin{gather}
		\label{eq:beta-ness}
		\upbeta(p_1) := \Theta(p_1) \beta_L + \Theta(-p_1)\beta_R
	\end{gather}
	for the Heaviside function $\Theta$ (indicator function of $(0,\infty)$). 
\end{definition}
A discussion that covers the technical details and physical interpretation may be found in \cite{Hack-Verch,DLSB}. 
Here, we only give a rough description of the construction in \cite{Hack-Verch}, and highlight some fundamental properties.\medskip

Let $\beta_L > 0$ and $\beta_R > 0$ be the inverse temperatures of two semi-infinite heat baths of the quantum field inside $x^1 < 0$ (``left'') and $x^1 > 0$ (``right''), respectively, with respect to the time evolution $\alpha$ along the $x^0$-axis (Eq.\ \eqref{eq:alpha}). 
We assume these heat baths to be in contact with each other in a neighborhood of the spatial surface $\{\xvec\in\R^3 : x^1 = 0\}$, with a smoothed-out transition region $\{\xvec\in\R^3 : |x^1|\leq\varepsilon \}$ for some $\varepsilon>0$. 
Together with an intermediate $\beta'$-KMS state on the transition region for some $\beta' \geq \max(\beta_L , \beta_R)$, the $\beta_L$- and $\beta_R$-KMS states of the heat baths are glued together by a construction using partitions of unity, which is presented in detail in \cite[Sec.\ 2.2]{Hack-Verch}. 
This results in a quasi-free Hadamard state $\omega_G$ by \cite[Thm.\ 3.1]{Hack-Verch}. 
(For simplicity, the chemical potentials of the initial semi-infinite heat baths and of the transition region between them are set to $0$, so our situation falls into the third case of \cite[(19)]{Hack-Verch}.) 
If $\WR := \{x\in\M : |x^0|<x^1 \}$ is the right wedge and $\esf_1$ is the spacelike unit vector along the $x^1$-axis, then $\omega_G$ restricts to a $\beta_L$-KMS state on the local observable algebra $\Acal (-\WR-\varepsilon \esf_1)$ and to a $\beta_R$-KMS state on $\Acal (\WR+\varepsilon \esf_1)$, with respect to time evolution $\alpha$. 
By \cite[Thm.\ 3.2 (1)]{Hack-Verch} the glued initial state evolves into the $\alpha$-invariant, quasi-free NESS
\begin{gather*}
	\omega_N = \lim_{t\to\infty} \omega_G \circ \alpha_t
\end{gather*}
on $\Acal$, completely determined by the two-point function $W_N$ specified in Eq.\ \eqref{eq:ness-two-point}. 
The NESS $\omega_N$ resulting from this construction does not depend on the KMS parameter $\beta'$ of the transition region (and its chemical potential) as long as the conditions laid out in \cite[Sec.\ 3.1]{Hack-Verch} are fulfilled (see \cite[Thm.\ 3.2 (1)]{Hack-Verch}). 
As noted in \cite{Hack-Verch}, the NESS agrees with the one constructed in \cite{DLSB}, which assumed the initial heat baths to be in sharp contact. 
The smooth transition region in the initial configuration of the system described above guarantees an initial state that satisfies the Hadamard condition everywhere, which enables the construction of NESS in interacting Klein-Gordon models \cite{Hack-Verch}.\medskip 

Some fundamental properties of the NESS $\omega_N$ can be inferred from the form of the two-point function $W_N$ in Eq.\ \eqref{eq:ness-two-point}. 
The state $\omega_N$ is translation-invariant in time and spatial directions, and it is a Hadamard state, i.e.\ $\widetilde{W}_N := W_N - W_{\mathrm{vac}}$ is a smooth function on $\M\times\M$ \cite{Hack-Verch}, given by (using Eq.\ \eqref{eq:vac-two-point}) 
\begin{gather}
	\label{eq:ness-two-point-reg}
	\widetilde{W}_N (x,y) = \frac{1}{(2\pi)^3} \int\limits_{\R^3} \frac{1}{2\lVert\pvec\rVert} \e^{i\pvec\cdot(\xvec-\yvec)} \left( \frac{\e^{i\lVert\pvec\rVert(x^0 - y^0)}}{\e^{\upbeta(-p_1)\lVert\pvec\rVert} - 1} + \frac{\e^{-i\lVert\pvec\rVert(x^0 - y^0)}}{\e^{\upbeta(p_1)\lVert\pvec\rVert} - 1} \right) \, \Diff3\pvec \, .
\end{gather}
Since $W_N$ coincides with the $\beta$-KMS two-point function $W_\beta$ (Eq.\ \eqref{eq:thermal-two-point}) if $\beta_L = \beta_R =: \beta$, we have to assume that $\beta_L \neq \beta_R$ to obtain a proper NESS. 
The two-point function then indicates that right-moving ($p_1 > 0$) and left-moving modes ($p_1 < 0$) along the $x^1$-axis are thermal at temperature $\beta_L^{-1}$ and $\beta_R^{-1}$, respectively \cite{DLSB,Hack-Verch}. 
This observation can be expressed rigorously as a ``modewise'' KMS property of the NESS (see \cite[Sec.\ 3.3]{Hack-Verch}). 
Moreover, the $(01)$ component of the expected stress-energy tensor is non-vanishing \cite{DLSB} (see Lemma \ref{lem:set-ness-components}). 
Hence the NESS features a heat (or energy) flow that travels from $x^1 = -\infty$ to $x^1 = +\infty$ for $\beta_L^{-1} > \beta_R^{-1}$ (and reversed for $\beta_L^{-1} < \beta_R^{-1}$).

\section{Transition rates in (non-)equilibrium states}
\label{sec:transition-rates}

We fix an inertial reference frame $I$ (the laboratory frame) corresponding to coordinates $(x^0 , x^1 , x^2 , x^3)$ on $\M\cong\R^4$. 
The inertial time evolution in this frame is implemented on the observable algebra $\Acal$ of the free massless scalar field by the $^\ast$-automorphism group $\alpha$ defined in Eq.\ \eqref{eq:alpha}. 
Let $I_{\vel}$ be the inertial reference frame that moves relative to $I\equiv I_0$ with constant velocity $\vel\in(-1,1)$ along the $x^1$-axis (which is the longitudinal axis of the NESS specified in Section \ref{sec:setup-ness}). 
In proper time $\tau\in\R$ the worldline of the detector at rest in $I_{\vel}$ (passing through the origin at $\tau=0$) is, expressed in $I$, 
\begin{gather}
	\label{eq:inertial-worldline}
	\tau\mapsto\xsf_{\vel} (\tau)=\gam (\tau , \vel \tau ,0,0) \, , \quad \gam:=\frac{1}{\sqrt{1-\vel^2}} \, .
\end{gather} 
We also define
\begin{gather*}
	\dop_\pm := \gam(1\pm\vel) = \sqrt{\frac{1\pm\vel}{1\mp\vel}} \, ,
\end{gather*}
which represent the Doppler factors. 

We consider the following two classes of scenarios for the detector-field system. 

\begin{enumerate}[align=left]
	\myitem[\textsf{(KMS-$\vel$)}] \label{item:KMS} The detector is stationary in $I_{\vel}$ and is coupled to a thermal equilibrium state relative to $I$ (the KMS state $\omega_\beta$ with respect to $\alpha$ at inverse temperature $\beta>0$, with two-point function $W_\beta$ given by Eq.\ \eqref{eq:thermal-two-point}).
	\myitem[\textsf{(NESS-$\vel$)}] \label{item:NESS} The detector is stationary in $I_{\vel}$ and is coupled to a NESS in $I$ (the NESS $\omega_N$ with respect to $\alpha$ and initial semi-infinite heat baths at inverse temperatures $\beta_L , \beta_R > 0$, with two-point function $W_N$ given by Eq.\ \eqref{eq:ness-two-point}).
\end{enumerate}

In order to calculate the transition rates in these scenarios it will be useful to split them into regular and distributional parts (as has been done in, e.g., \cite{Biermann-et-al2020,Bunney-Parry-Perche-Louko2024}). 
For the inertial worldlines we can exploit the Hadamard property of the initial field states by extracting the contribution of the Minkowski vacuum state (the ground state of the massless scalar field relative to every inertial frame), which has the two-point function $W_{\mathrm{vac}}$ given by Eq.\ \eqref{eq:vac-two-point}. 
A common expression for the vacuum two-point function along a worldline $\tau\mapsto\xsf(\tau)$ in terms of a formal, regularized integral kernel would be the $i\epsilon$-prescription $W_{\mathrm{vac}} (\xsf(\tau),\xsf(\tau')) = -\frac{1}{4\pi^2} \lim_{\epsilon\to 0_+} 1/\bigl((x^0 (\tau) - x^0 (\tau') - i\epsilon)^2 - \lVert\xvec(\tau)-\xvec(\tau')\rVert^2 \bigr)$ (see \cite[Eq.\ (3.59)]{Birrell-Davies}), in which it is implied that one has to integrate against test functions before taking the limit $\epsilon\to 0$ from above. 
An alternative, manifestly Lorentz-invariant regularization is $W_{\mathrm{vac}} (\xsf(\tau),\xsf(\tau')) = -\frac{1}{4\pi^2} \lim_{\epsilon\to 0_+} 1/\bigl( \xsf(\tau)-\xsf(\tau') -i\epsilon(\dot{\xsf}(\tau)+\dot{\xsf}(\tau')) \bigr)^2$ for the four-velocity $\dot{\xsf}$ \cite{Schlicht2004} (note that the global minus sign does not appear in \cite{Schlicht2004} due to the inverted choice of metric signature). 
Thus
\begin{gather*}
	W_{\mathrm{vac}} (\xsf_{\vel} (s), 0) = W_{\mathrm{vac}} (\xsf_0 (s), 0) = -\frac{1}{4\pi^2} \lim\limits_{\epsilon\to 0_+} \frac{1}{(s-i\epsilon)^2}
\end{gather*}
for every velocity $\vel$, and the vacuum transition rate $\Rcal_{\mathrm{vac}}$ of an inertially moving detector is shown to be
\begin{gather}
	\label{eq:transition-rate-vac}
	\Rcal_{\mathrm{vac}} (E) := \int\limits_{\R} \e^{-iEs} W_{\mathrm{vac}} (\xsf_{\vel} (s), 0) \diff s = -\frac{E}{2\pi} \Theta(-E)
\end{gather}
for $E\in\R\setminus\{0\}$ by contour integration (see, e.g., \cite{Schlicht2004}, and also \cite{Biermann-et-al2020,Bunney-Parry-Perche-Louko2024}). 
As expected, the transition rate vanishes for $E>0$, i.e.\ there is no detector excitation in the vacuum.

\subsection{Detector moving inertially through a heat bath}
\label{sec:case-KMS}

Let us first consider the scenario \ref{item:KMS}. 
If the field is prepared in the $\beta$-KMS state $\omega_\beta$ and the detector moves along the inertial worldline $\xsf_{\vel}$ with constant non-zero velocity $\vel$, the transition rate is \cite{Costa-Matsas-background1995,Costa-Matsas1995}
\begin{gather}
	\label{eq:transition-rate-kms-moving}
	\Rcal_{\vel}^{(\beta)} (E) = \frac{1}{4\pi\beta\gam\vel} \, \ln\left(\frac{1-\e^{-\beta \dop_+ |E|}}{1-\e^{-\beta \dop_- |E|}} \right) - \frac{E}{2\pi} \Theta(-E) \equiv \frac{1}{4\pi\beta\gam\vel} \, \ln\left(\frac{1-\e^{-\beta \dop_+ E}}{1-\e^{-\beta \dop_- E}} \right)
\end{gather}
for $E\in\R\setminus\{0\}$, where we write $\Rcal_{\vel}^{(\beta)} := \Rcal_{\omega_\beta}$ for the transition rate (defined in Eq.\ \eqref{eq:transition-rate}) to highlight the dependence on $\vel$. 
The expression reveals the relativistic Doppler shift \cite{Rindler1977} of the radiation perceived by the detector. 
It holds $\Rcal_{-\vel}^{(\beta)} = \Rcal_{\vel}^{(\beta)}$ for $0<|\vel|<1$, which is to be expected since the KMS state $\omega_\beta$ is homogeneous and isotropic and thus the (direction-independent) transition rate of the monopole detector should not depend on the direction of the linear motion. 
We provide a derivation of Eq.\ \eqref{eq:transition-rate-kms-moving} in Appendix \ref{appendix:cm}. 
For small energy gaps, 
\begin{gather}
	\label{eq:transition-rate-kms-moving-small-gap}
	\Rcal_{\vel}^{(\beta)} (E) = \frac{\sqrt{1-\vel^2}}{4\pi\beta\vel} \ln\left(\frac{1+\vel}{1-\vel}\right) - \frac{E}{4\pi} + \mathcal{O}(E^2) \, .
\end{gather}

Some properties of the transition rate for $E>0$ can be extracted from Figure \ref{fig:rate-kms-moving}. 
\begin{figure}[!t]
	\centering
	\begin{subfigure}{.5\textwidth}
		\centering
		\includegraphics[width=\textwidth]{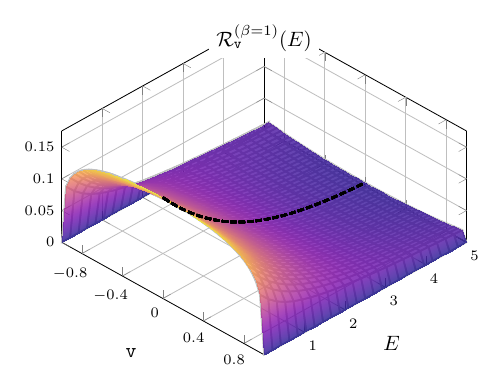}
		\caption{$(\vel,E)\mapsto\Rcal_{\vel}^{(\beta=1)} (E)$}
		\label{fig:sub1-kms-moving}
	\end{subfigure}%
	\begin{subfigure}{.5\textwidth}
		\centering
		\includegraphics[width=\textwidth]{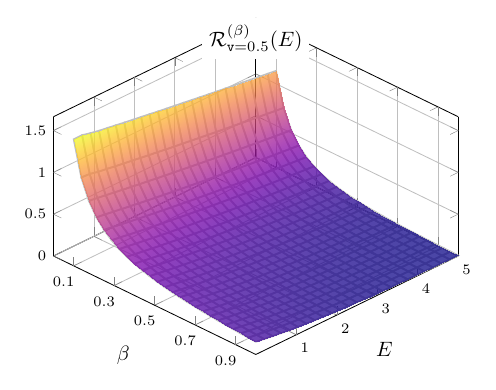}
		\caption{$(\beta,E)\mapsto\Rcal_{\vel=0.5}^{(\beta)} (E)$ ($\beta>0.1$)}
		\label{fig:sub2-kms-moving}
	\end{subfigure}
	\caption{\textbf{(a)} Surface plot of $(\vel,E)\mapsto \Rcal_{\vel}^{(\beta)} (E)$ for fixed $\beta=1$ and $\vel\in(-1,1)$, $E>0$. The transition rate $\Rcal_{\vel}^{(\beta)} (E)$ is given by Eq.\ \eqref{eq:transition-rate-kms-moving} and is even in $\vel$, i.e.\ $\Rcal_{-\vel}^{(\beta)} = \Rcal_{\vel}^{(\beta)}$. For non-negative $\vel$ the plot reproduces Fig.\ 1 in \cite{Costa-Matsas1995}. The maximum of $\vel\mapsto\Rcal_{\vel}^{(\beta)} (E)$ shifts from $\vel=0$ for $E\ll\beta^{-1} = 1$ to non-zero $|\vel|>0$ for $E\gg\beta^{-1} = 1$; see \cite{Costa-Matsas1995} for a discussion (and also the dashed lines in Figure \ref{fig:sub2-ness-moving-infty}, where this can be seen more clearly). The dashed surface curve at $\vel=0$ represents the transition rate of a stationary detector and is given by the Planckian $\Rcal_{0}^{(\beta=1)}$ (Eq.\ \eqref{eq:transition-rate-kms-rest}). \quad \textbf{(b)} Surface plot of $(\beta,E)\mapsto \Rcal_{\vel}^{(\beta)} (E)$ for fixed $\vel=0.5$ and $\beta>0.1$, $E>0$.}
	\label{fig:rate-kms-moving}
\end{figure}
In particular, as the velocity increases, the Doppler shift eventually prevents the response of the detector, so that $\Rcal_{\vel}^{(\beta)}$ vanishes pointwise as $|\vel|\to 1$ \cite{Costa-Matsas1995}. 
(For $E<0$ one obtains the vacuum contribution $-\frac{E}{2\pi}$ in this limit; see Eq.\ \eqref{eq:transition-rate-kms-moving}.) 
In the context of an experimental verification of the Unruh effect \cite{Unruh1976,Crispino-Higuchi-Matsas} this contributes to the conclusion that for a uniformly accelerating detector in a background thermal bath the influence of the latter is suppressed for sufficiently high accelerations \cite{Costa-Matsas-background1995}. 
The non-Planckian form of $\Rcal_{\vel}^{(\beta)}$ has been argued to prove the lack of a continuous Lorentz transformation law for temperature \cite{Landsberg-Matsas1996,Landsberg-Matsas2004} (which found a model-independent corroboration at the level of KMS states in \cite{Sewell2008,Sewell-rep2009}; see also \cite{PV-disj} for related remarks).\medskip

In the limit $\vel\to 0$, approaching the case \textsf{(KMS-$0$)} of a detector at rest relative to the reference frame of the thermal equilibrium state, one obtains the well-known transition rate
\begin{gather}
	\Rcal_{0}^{(\beta)} (E) = \lim\limits_{\vel\to 0} \Rcal_{\vel}^{(\beta)} (E) = \frac{1}{2\pi} \frac{|E|}{\e^{\beta |E|} -1} - \frac{E}{2\pi} \Theta(-E) \equiv \frac{1}{2\pi} \frac{E}{\e^{\beta E} -1}
	\label{eq:transition-rate-kms-rest}
\end{gather}
corresponding to a Planck (black body) spectrum at the temperature $\beta^{-1}$ \cite{Costa-Matsas-background1995,Costa-Matsas1995,Costa2004}. 
Notice that in the zero temperature limit $\beta\to\infty$ the vacuum transition rate (Eq.\ \eqref{eq:transition-rate-vac}) is recovered.

\subsection{Stationary detector coupled to a NESS}
\label{sec:case-NESS-stat}

We now turn to the scenario \ref{item:NESS} for $\beta_L , \beta_R >0$. 
From Eq.\ \eqref{eq:ness-two-point} we notice that $W_N (\xsf_{\vel} (\tau),\xsf_{\vel} (\tau')) = W_N (\xsf_{\vel} (\tau-\tau'),0)$ for all $\tau,\tau'\in\R$ (so Eq.\ \eqref{eq:transition-rate} for the transition rate applies). 
Define 
\begin{gather*}
	w_{N,\vel} (s):=W_N (\xsf_{\vel} (s),0) \, , \quad \widetilde{w}_{N,\vel} (s) := \widetilde{W}_N (\xsf_{\vel} (s),0) \equiv (W_N - W_{\mathrm{vac}}) (\xsf_{\vel} (s),0)
\end{gather*}
for the vacuum two-point function $W_{\mathrm{vac}}$, where $\widetilde{w}_{N,\vel}$ is a smooth function on $\R$ by the Hadamard property of the NESS $\omega_N$ (see Section \ref{sec:setup-ness}) with $\widetilde{W}_N$ given by Eq.\ \eqref{eq:ness-two-point-reg}. 
Let us first consider \textsf{(NESS-$0$)}, i.e.\ the detector is at rest relative to the frame $I$ of the NESS $\omega_N$ to which it is coupled. 

The modewise thermal interpretation of the NESS mentioned in Section \ref{sec:setup-ness} allows to infer that right- and left-moving modes (with momenta $p_1 > 0$ and $p_1 < 0$, respectively) each contribute half of a Planckian to the transition rate (Eq.\ \eqref{eq:transition-rate}), so the transition rate of a detector at rest in the NESS is given by an average of thermal transition rates corresponding to the inverse temperatures $\beta_L , \beta_R$. 
For completeness we provide a precise proof for this observation. 
\begin{proposition}[Transition rate of detector at rest in NESS]
	\label{prop:ness-stat}
	The transition rate for \textup{\textsf{(NESS-$0$)}} (with $\beta_L , \beta_R >0$) is given by
	\begin{gather}
		\label{eq:transition-rate-ness-stat}
		\Rcal_{N,0}^{(\beta_L , \beta_R)} (E) := \int\limits_{\R} \e^{-iEs} w_{N,0} (s) \diff s = \frac{1}{2} \left(\Rcal_{0}^{(\beta_L)} (E) + \Rcal_{0}^{(\beta_R)} (E)\right) \, ,
	\end{gather}
	which is the arithmetic mean of the Planckian transition rates (Eq.\ \eqref{eq:transition-rate-kms-rest}) corresponding to the inverse temperatures $\beta_L , \beta_R$ of the two initial semi-infinite heat baths.
\end{proposition}
\begin{proof}
	From Eq.\ \eqref{eq:ness-two-point-reg} we get
	\begin{gather*}
		\widetilde{w}_{N,0} (s)=\frac{1}{(2\pi)^3} \int\limits_{\R^3} \frac{1}{2\lVert\pvec\rVert} \left( \frac{\e^{i\lVert\pvec\rVert s}}{\e^{\upbeta(-p_1)\lVert\pvec\rVert} - 1} + \frac{\e^{-i\lVert\pvec\rVert s}}{\e^{\upbeta(p_1)\lVert\pvec\rVert} - 1} \right) \, \Diff3\pvec
	\end{gather*}
	for $s\in\R$. 
	This boils down to the sum of two integrals of the form $\int_0^\infty \e^{irs} \frac{r}{\e^{\beta r} - 1} \diff r$, which converges for all $\beta>0$ and $s\in\R$ (see \cite[Sec.\ 1.4, (8)]{Erdelyi} and \cite[Sec.\ 2.3, 3.13]{Oberhettinger-Fourier}). 
	Substituting $p_1 \mapsto -p_1$ in the first term one finds that
	\begin{gather}
		\widetilde{w}_{N,0} (s) = \frac{1}{(2\pi)^3} \int\limits_{\R^3} \frac{1}{\lVert\pvec\rVert} \frac{\cos(\lVert\pvec\rVert s)}{\e^{\upbeta(p_1)\lVert\pvec\rVert} - 1} \, \Diff3\pvec = \frac{1}{4\pi^2} \left( C_{\beta_L} (s) + C_{\beta_R} (s) \right) \label{eq:ness-two-point-stat}
	\end{gather}
	by Eq.\ \eqref{eq:beta-ness} and using spherical coordinates on the half-spaces with $p_1 > 0$ and $p_1 < 0$, where 
	\begin{gather}
		\label{eq:Cbeta}
		C_\beta (s) := \int_0^\infty \frac{r\cos(rs)}{\e^{\beta r} - 1} \diff r = \frac{1}{2s^2} - \frac{\pi^2}{2\beta^2} \frac{1}{\sinh^2 \left(\frac{\pi s}{\beta}\right)} \, , \quad \beta > 0 \, , \, s\in\R
	\end{gather}
	by \cite[Sec.\ 1.4, (8)]{Erdelyi}, understood as smooth function on $\R$ with $\lim_{s\to 0} C_\beta (s) = \frac{\pi^2}{6\beta^2} = C_\beta (0)$ (see \cite[Sec.\ 6.3, (7)]{Erdelyi}). 
	Writing $w_{N,0} (s) = \widetilde{w}_{N,0} (s) + W_{\mathrm{vac}} (\xsf_{0} (s),0)$, we have 
	\begin{gather*}
		\Rcal_{N,0}^{(\beta_L , \beta_R)} (E) = 2 \int_0^\infty \cos(|E|s) \widetilde{w}_{N,0} (s) \diff s - \frac{E}{2\pi} \Theta(-E)
	\end{gather*}
	by Eq.\ \eqref{eq:transition-rate-vac}, where we used that the function $\widetilde{w}_{N,0}$ is even and thus the exponential in the Fourier transform can be replaced by $\cos(Es)=\cos(|E|s)$. 
	Since $C_\beta \in L^1 (\R)$ the Fourier cosine transform in Eq.\ \eqref{eq:Cbeta} can be inverted to $\int_0^\infty \cos(qs) C_\beta (s) \diff s = \frac{\pi}{2} q(\e^{\beta q} - 1)^{-1}$ for $\beta,q>0$. 
	Hence
	\begin{gather*}
		\Rcal_{N,0}^{(\beta_L , \beta_R)} (E) = \frac{1}{4\pi} \left( \frac{|E|}{\e^{\beta_L |E|}-1} + \frac{|E|}{\e^{\beta_R |E|}-1} \right) - \frac{E}{2\pi} \Theta(-E)
	\end{gather*}
	and the result follows. 
\end{proof}

\begin{figure}[!t]
	\centering
	\begin{subfigure}{.5\textwidth}
		\centering
		\includegraphics[width=\textwidth]{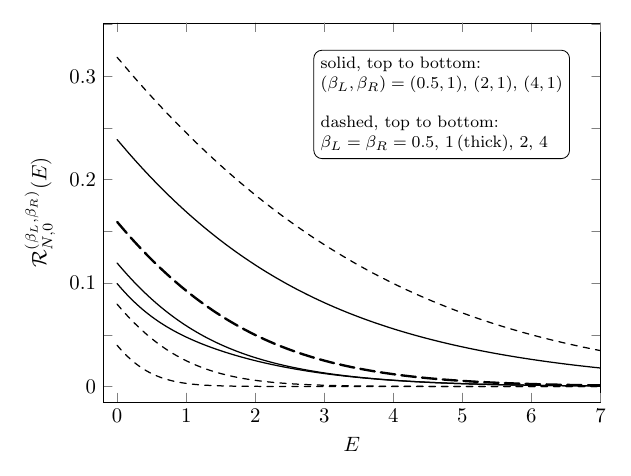}
		\caption{$\Rcal_{N,0}^{(\beta_L , \beta_R = 1)}$ (solid) and $\Rcal_{0}^{(\beta)}$ (dashed)}
		\label{fig:sub1-ness-kms-stationary}
	\end{subfigure}%
	\begin{subfigure}{.5\textwidth}
		\centering
		\includegraphics[width=\textwidth]{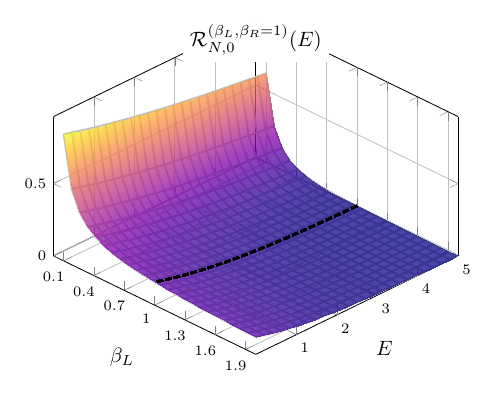}
		\caption{$(\beta_L , E) \mapsto \Rcal_{N,0}^{(\beta_L , \beta_R = 1)} (E)$ ($\beta_L > 0.1$)}
		\label{fig:sub2-ness-stationary}
	\end{subfigure}
	\caption{\textbf{(a)} Plot of the transition rate $\Rcal_{N,0}^{(\beta_L , \beta_R)}$, given by Eq.\ \eqref{eq:transition-rate-ness-stat}, for $E>0$, fixed $\beta_R=1$, and different values of $\beta_L$ (solid), in comparison with the Planckian transition rates $\Rcal_{0}^{(\beta)}$ (Eq.\ \eqref{eq:transition-rate-kms-rest}) for $\beta$ from the same set as $\beta_L$ (dashed). The rate $\Rcal_{0}^{(\beta=1)}$, which equals $\Rcal_{N,0}^{(\beta_L = 1 , \beta_R = 1)}$, is highlighted as thick dashed line. \quad \textbf{(b)} Surface plot of $(\beta_L , E) \mapsto \Rcal_{N,0}^{(\beta_L , \beta_R)} (E)$ for fixed $\beta_R=1$ and $\beta_L > 0.1$, $E>0$. The dashed surface curve at $\beta_L = 1$ represents the Planckian $\Rcal_{0}^{(\beta=1)}$.}
	\label{fig:rate-ness-stationary}
\end{figure}

For $\beta_L = \beta_R =: \beta$ the transition rate $\Rcal_{N,0}^{(\beta_L , \beta_R)}$ equals the Planckian $\Rcal_{0}^{(\beta)}$, as expected. 
For a proper NESS with $\beta_L \neq \beta_R$, however, the detector is excited by modes coming from $x^1 = -\infty$ and $x^1 = \infty$ with Planck spectrum of temperature $\beta_L^{-1}$ and $\beta_R^{-1}$, respectively. 
The detector, which is at rest relative to the two heat baths, effectively reacts to a mixture of the corresponding KMS states. 
If the transition spectrum is sampled for a sufficiently large set of detector energies $E$, one will always find a non-thermal detector response (i.e.\ not corresponding to a Planck spectrum). 
Figure \ref{fig:rate-ness-stationary} shows plots of $\Rcal_{N,0}^{(\beta_L , \beta_R)}$ for fixed $\beta_R = 1$ and different values of $\beta_L$ when $E>0$. 
Irrespective of $\beta_L , \beta_R$, the transition rate vanishes as $E\to\infty$ and is asymptotic to the function $E\mapsto -\frac{E}{2\pi}$ as $E\to -\infty$, i.e.\ $\lim_{E\to -\infty} (\Rcal_{N,0}^{(\beta_L , \beta_R)} (E) + \frac{E}{2\pi})=0$.

\subsection{Detector moving inertially through a NESS}
\label{sec:case-NESS-vel}

Consider the case \ref{item:NESS} for $\vel\in(-1,1)\setminus\{0\}$, that is, the detector couples to $\omega_N$ and travels along the inertial worldline $\xsf_{\vel}$ (Eq.\ \eqref{eq:inertial-worldline}) with constant non-zero velocity $\vel$. 

\begin{proposition}[Transition rate of detector moving inertially through NESS]
	\label{prop:ness-v}
	The transition rate for \textup{\textsf{(NESS-$\vel$)}} (with $\beta_L , \beta_R >0$ and $\vel\in(-1,1)\setminus\{0\}$) is given by
	\begin{gather}
		\label{eq:transition-rate-ness-vel}
		\Rcal_{N,\vel}^{(\beta_L , \beta_R)} (E) := \int\limits_{\R} \e^{-iEs} w_{N,\vel} (s) \diff s = R(\beta_L , -\vel , E) + R(\beta_R , \vel , E) \, , 
	\end{gather}
	where
	\begin{gather}
		\label{eq:R-function}
		R(\beta,\vel,E) := \frac{\sqrt{1-\vel^2}}{4\pi\vel\beta} \ln\left(\frac{1-\e^{-\beta \sqrt{1-\vel^2} E}}{1-\e^{-\beta \sqrt{\frac{1-\vel}{1+\vel}} E}} \right)
	\end{gather}
	for $\beta>0$, $\vel\in(-1,1)\setminus\{0\}$, $E\in\R\setminus\{0\}$. 
	Expanded in small $|E|$, 
	\begin{gather}
		\label{eq:transition-rate-ness-vel-small-gap}
		\Rcal_{N,\vel}^{(\beta_L , \beta_R)} (E) = \frac{\sqrt{1-\vel^2}}{4\pi\vel} \left( \frac{1}{\beta_R} \ln(1+\vel) - \frac{1}{\beta_L} \ln(1-\vel) \right) -\frac{E}{4\pi} + \mathcal{O}(E^2) \, .
	\end{gather}
\end{proposition}
\begin{proof}
	As $\pvec\cdot\vec{\xsf_{\vel} (s)} = p_1 \gam\vel s$, Eq.\ \eqref{eq:ness-two-point-reg} implies
	\begin{gather}
		\label{eq:ness-two-point-vel}
		\widetilde{w}_{N,\vel} (s) = \widetilde{W}_N (\xsf_{\vel} (s),0) = \frac{1}{(2\pi)^3} \int\limits_{\R^3} \frac{1}{2\lVert\pvec\rVert} \left( \frac{\e^{ip_1 \gam \vel s} \e^{i\lVert\pvec\rVert \gam s}}{\e^{\upbeta(-p_1)\lVert\pvec\rVert} - 1} + \frac{\e^{ip_1 \gam \vel s} \e^{-i\lVert\pvec\rVert \gam s}}{\e^{\upbeta(p_1)\lVert\pvec\rVert} - 1} \right) \, \Diff3\pvec \, ,
	\end{gather}
	and substituting $p_1 \mapsto -p_1$ in the first term yields
	\begin{gather}
		\label{eq:ness-two-point-vel-1}
		\widetilde{w}_{N,\vel} (s) = \frac{1}{(2\pi)^3} \int\limits_{\R^3} \frac{1}{\lVert\pvec\rVert} \frac{\cos\left(p_1 \gam \vel s - \lVert\pvec\rVert \gam s \right)}{\e^{\upbeta(p_1)\lVert\pvec\rVert} - 1} \, \Diff3\pvec \, .
	\end{gather}
	This integral is split into a sum according to the definition of $\upbeta(p_1)$ (Eq.\ \eqref{eq:beta-ness}). 
	Choosing spherical coordinates with $\theta$ the polar angle from the polar axis placed along the positive $x^1$-axis one finds by Eq.\ \eqref{eq:sine-transf}
	\begin{align}
		\int\limits_{\{p_1 > 0\} \times \R^2} \frac{1}{\lVert\pvec\rVert} & \frac{\cos\left(p_1 \gam \vel s - \lVert\pvec\rVert \gam s \right)}{\e^{\upbeta(p_1)\lVert\pvec\rVert} - 1} \, \Diff3\pvec = \nonumber \\ & = 2\pi \int_0^\infty \int_0^{\pi/2} \frac{r \cos\left(\gam \vel s r\cos(\theta) - \gam s r \right)}{\e^{\beta_L r} - 1} \, \sin(\theta) \diff\theta \diff r = \nonumber \\ & = -\frac{2\pi}{\gam\vel s} \int_0^\infty \frac{\sin\left(\gam(1-\vel)sr\right)-\sin(\gam sr)}{\e^{\beta_L r} - 1} \, \diff r = \nonumber \\ & = -\frac{2\pi}{\gam\vel s} \left[ -\frac{\dop_+}{2s} + \frac{\pi}{2\beta_L} \coth\left( \frac{\pi s}{\beta_L \dop_+} \right) + \frac{1}{2\gam s} - \frac{\pi}{2\beta_L} \coth\left( \frac{\pi\gam s}{\beta_L} \right) \right] \label{eq:ness-vel-1}
	\end{align}
	for the Doppler factors $\dop_\pm = \gam(1\pm\vel)$. 
	This equation holds for all $s\in\R$, with the last expression smoothly extending to $s=0$. 
	Similarly,
	\begin{align}
		\int\limits_{\{p_1 < 0\} \times \R^2} \frac{1}{\lVert\pvec\rVert} & \frac{\cos\left(p_1 \gam \vel s - \lVert\pvec\rVert \gam s \right)}{\e^{\upbeta(p_1)\lVert\pvec\rVert} - 1} \, \Diff3\pvec = \nonumber \\ & = 2\pi \int_0^\infty \int_{\pi/2}^{\pi} \frac{r \cos\left(\gam \vel s r\cos(\theta) - \gam s r \right)}{\e^{\beta_R r} - 1} \, \sin(\theta) \diff\theta \diff r = \nonumber \\ & = \frac{2\pi}{\gam\vel s} \left[ -\frac{\dop_-}{2s} + \frac{\pi}{2\beta_R} \coth\left( \frac{\pi s}{\beta_R \dop_-} \right) + \frac{1}{2\gam s} - \frac{\pi}{2\beta_R} \coth\left( \frac{\pi\gam s}{\beta_R} \right) \right] \, . \label{eq:ness-vel-2}
	\end{align}
	Taking the sum of \eqref{eq:ness-vel-1} and \eqref{eq:ness-vel-2} and using $\dop_+ - \dop_- = 2\gam\vel$, Eq.\ \eqref{eq:ness-two-point-vel-1} becomes
	\begin{align}
		&\widetilde{w}_{N,\vel} (s) = \frac{1}{4\pi^2 s^2} + \nonumber \\ & + \frac{1}{8\pi\gam\vel\beta_R s} \left[ \coth\left( \frac{\pi s}{\beta_R \dop_-} \right) - \coth\left( \frac{\pi\gam s}{\beta_R} \right) \right] - \frac{1}{8\pi\gam\vel\beta_L s} \left[ \coth\left( \frac{\pi s}{\beta_L \dop_+} \right) - \coth\left( \frac{\pi\gam s}{\beta_L} \right) \right] = \nonumber \\ & = \frac{1}{4\pi\gam\vel} \left[ \frac{1}{\beta_R} \left( f_{\beta_R / \gam} (s) - f_{\beta_R \dop_-} (s) \right) - \frac{1}{\beta_L} \left( f_{\beta_L / \gam} (s) - f_{\beta_L \dop_+} (s) \right) \right]
		\label{eq:ness-two-point-vel-2}
	\end{align}
	for the smooth function $f_c$ ($c>0$) defined in Eq.\ \eqref{eq:fc}. 
	An alternative derivation is presented in Appendix \ref{appendix:NESS-v}. 
	For $\beta_L = \beta_R =: \beta > 0$, Eq.\ \eqref{eq:ness-two-point-vel-2} reduces to Eq.\ \eqref{eq:two-point-reg} (case \textsf{(KMS-$\vel$)}). 
	Furthermore, for all $\beta_L , \beta_R > 0$, one regains Eq.\ \eqref{eq:ness-two-point-stat} in the limit $\vel\to 0$ (case \textsf{(NESS-$0$)}). 
	Using Eq.\ \eqref{eq:transition-rate-vac} together with Eq.\ \eqref{eq:ness-two-point-vel-2} and the Fourier transform of the function $f_c$ ($c>0$) given in Eq.\ \eqref{eq:fc-fourier}, we obtain
	\begin{align*}
		\Rcal_{N,\vel}^{(\beta_L , \beta_R)} (E) &= \int\limits_{\R} \e^{-iEs} \widetilde{w}_{N,\vel} (s) \diff s - \frac{E}{2\pi} \Theta(-E) = \\ &= \frac{1}{4\pi\gam\vel} \left[ \frac{1}{\beta_R} \ln\left(\frac{1-\e^{-\beta_R |E|/\gam}}{1-\e^{-\beta_R \dop_- |E|}} \right) - \frac{1}{\beta_L} \ln\left(\frac{1-\e^{-\beta_L |E|/\gam}}{1-\e^{-\beta_L \dop_+ |E|}} \right) \right] - \frac{E}{2\pi} \Theta(-E) = \\ &= \frac{1}{4\pi\gam\vel} \left[ \frac{1}{\beta_R} \ln\left(\frac{1-\e^{-\beta_R E/\gam}}{1-\e^{-\beta_R \dop_- E}} \right) - \frac{1}{\beta_L} \ln\left(\frac{1-\e^{-\beta_L E/\gam}}{1-\e^{-\beta_L \dop_+ E}} \right) \right] \, ,
	\end{align*}
	where the Heaviside function term is absorbed by some rearrangements.
\end{proof}

The motion of the detector along the $x^1$-axis with velocity $\vel\neq 0$ relative to the two semi-infinite heat baths results in a transition rate $\Rcal_{N,\vel}^{(\beta_L , \beta_R)}$ that deviates from the mean transition rate $\Rcal_{N,0}^{(\beta_L , \beta_R)} = \lim_{\vel\to 0} \Rcal_{N,\vel}^{(\beta_L , \beta_R)}$ corresponding to the individual $\beta_L$- and $\beta_R$-KMS states (Eq.\ \eqref{eq:transition-rate-ness-stat}). 
One can interpret this behavior to be similar to the case \textsf{(KMS-$\vel$)} of a detector moving with constant velocity relative to an inertial KMS state (Section \ref{sec:case-KMS}). 
In contrast to $\Rcal_{\vel}^{(\beta)}$ (Eq.\ \eqref{eq:transition-rate-kms-moving}), however, the transition rate $\Rcal_{N,\vel}^{(\beta_L , \beta_R)}$ is not an even function of $\vel$ unless $\beta_L = \beta_R =: \beta$, in which case the detector interacts with a global $\beta$-KMS state and the transition rate reduces to $\Rcal_{N,\vel}^{(\beta , \beta)} = \Rcal_{\vel}^{(\beta)}$. 
Moreover, $\Rcal_{N,\vel}^{(\beta_L , \beta_R)}$ is not invariant under the exchange of $\beta_L$ and $\beta_R$ for non-zero detector velocity and $\beta_L \neq \beta_R$. 
Hence we find that a relation analogous to Eq.\ \eqref{eq:transition-rate-ness-stat} does not hold, to wit,
\begin{gather*}
	\Rcal_{N,\vel}^{(\beta_L , \beta_R)} \neq \frac{1}{2} \left(\Rcal_{\vel}^{(\beta_L)} + \Rcal_{\vel}^{(\beta_R)} \right) \, , \quad \vel\in(-1,1)\setminus\{0\} \, , \quad \beta_L \neq \beta_R
\end{gather*}
for $\Rcal_{\vel}^{(\beta)}$ given by Eq.\ \eqref{eq:transition-rate-kms-moving}. 
However, the transition rate $\Rcal_{N,\vel}^{(\beta_L , \beta_R)}$ is invariant under simultaneously exchanging the $\beta_L$- and $\beta_R$-KMS heat baths as well as reversing the direction of motion along the $x^1$-axis, i.e.\ $\vel\mapsto -\vel$. 
This behavior is to be expected, because the detector's motion relative to the semi-infinite heat baths distinguishes a direction (determined by the sign of $\vel$) in relation to the thermal gradient of the NESS (determined by the sign of $\beta_R - \beta_L$), manifesting in a Doppler shift as signified by the Doppler factors appearing in Eq.\ \eqref{eq:transition-rate-ness-vel}. 
No such direction is distinguished in the stationary case \textsf{(NESS-$0$)}; the transition rate incorporates quanta incoming from all spatial directions, making $\Rcal_{N,0}^{(\beta_L , \beta_R)}$ symmetric in $\beta_L , \beta_R$ despite the presence of a heat flow.\medskip

\begin{figure}[!t]
	\centering
	\begin{subfigure}{.5\textwidth}
		\centering
		\includegraphics[width=\textwidth]{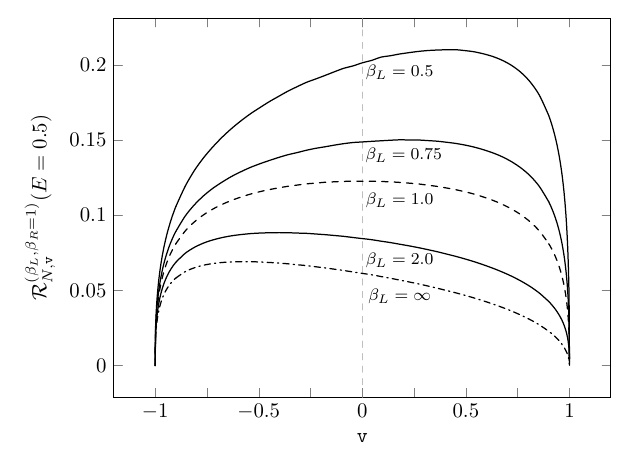}
	\end{subfigure}%
	\begin{subfigure}{.5\textwidth}
		\centering
		\includegraphics[width=\textwidth]{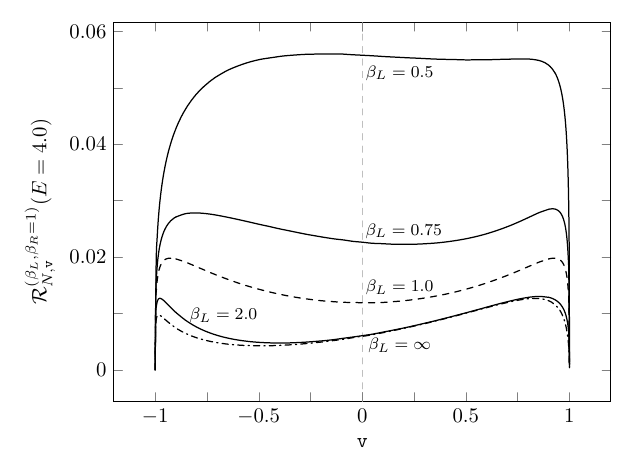}
	\end{subfigure}
	\caption{Plot of $\Rcal_{N,\vel}^{(\beta_L , \beta_R)} (E)$ as a function of $\vel$ for $\beta_R=1$, $E=0.5$ (left) and $E=4.0$ (right), and different values of $\beta_L$. In both plots the dashed line represents $\beta_L = \beta_R = 1$ and thus the transition rate $\Rcal_{\vel}^{(\beta=1)}$ for the $(\beta=1)$-KMS state (Eq.\ \eqref{eq:transition-rate-kms-moving}), and the case $\beta_L = \infty$ is shown by a dash-dotted line.}
	\label{fig:rate-ness-moving}
\end{figure}

For the rest of this section we restrict our attention to positive energy gaps $E>0$. 
Figure \ref{fig:rate-ness-moving} shows the transition rate $\Rcal_{N,\vel}^{(\beta_L , \beta_R)}$ as a function of $\vel$ for arbitrarily chosen values $\beta_R = 1$, $\beta_L \in\{0.5,0.75,1,2,\infty\}$, $E\in\{0.5,4\}$. 
First of all, we see that $\Rcal_{N,\vel}^{(\beta_L , \beta_R)}$ vanishes pointwise for $|\vel|\to 1$. 
Following the explanation given in \cite{Costa-Matsas-background1995,Costa-Matsas1995} for a detector moving in a heat bath, we can interpret this in terms of the Doppler effect gradually shifting the quanta of the surrounding NESS outside the detectable range as $|\vel|$ approaches the speed of light, thereby impeding the excitation of the detector. 
The plots also reveal the asymmetry in $\vel$ for $\beta_L \neq \beta_R$ and fixed $E$, which reflects the asymmetry under the exchange of the heat baths as noted above. 
At certain velocities the Doppler shift facilitates the excitation of the detector: 
For $E=0.5$ the maximal transition rate is attained for some $\vel>0$ when $\beta_L < \beta_R$, and for some $\vel<0$ when $\beta_L > \beta_R$. 
For larger energy gaps, exemplified by the value $E=4.0$ in Figure \ref{fig:rate-ness-moving}, the transition rate admits two maxima located towards large $|\vel|$.

On physical grounds we would expect a combination of Doppler shifts taking effect in the rest frame of the detector, one caused by the detector's motion relative to the laboratory frame, and another caused by the heat flow of the NESS. 
For $E=0.5$ we see that $\Rcal_{N,\vel}^{(\beta_L , \beta_R)} (E)$ peaks when the detector moves in the direction of the heat flow. 
Considering the isotropic response in the stationary case \textsf{(NESS-$0$)}, we can take the behavior of $\Rcal_{N,\vel}^{(\beta_L , \beta_R)}$ as an indicator of purely kinematical effects from the presence of the two heat baths. 
Notice that the asymmetry in $\vel$ subsides when both semi-infinite heat baths have the same temperature, i.e.\ when the field is in an inertial KMS state, for which the transition rate takes its maximum at $\vel=0$ for sufficiently small $E>0$ (see the discussion in \cite{Costa-Matsas-background1995,Costa-Matsas1995}).\medskip 

\begin{figure}[!t]
	\centering
	\begin{subfigure}{.5\textwidth}
		\centering
		\includegraphics[width=\textwidth]{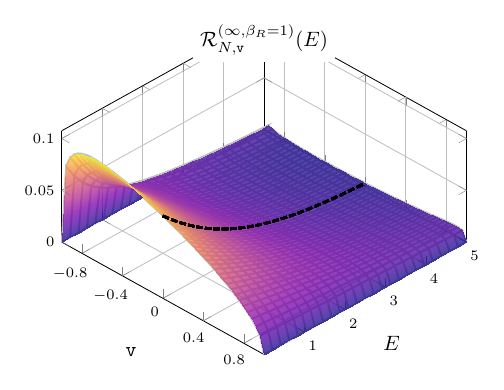}
		\caption{$(\vel,E)\mapsto\Rcal_{N,\vel}^{(\infty , \beta_R = 1)} (E)$}
		\label{fig:sub1-ness-moving-infty}
	\end{subfigure}%
	\begin{subfigure}{.5\textwidth}
		\centering
		\includegraphics[width=\textwidth]{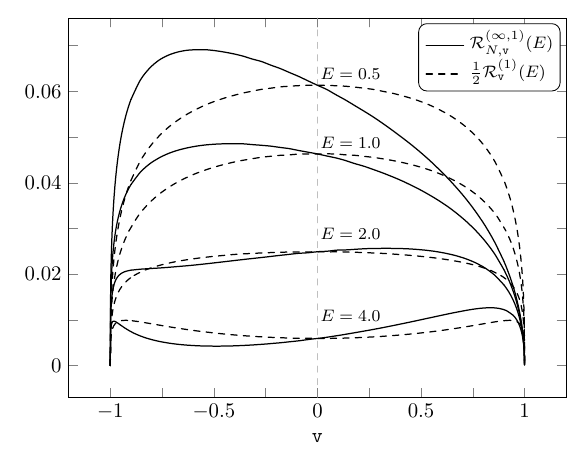}
		\caption{$\vel\mapsto\Rcal_{N,\vel}^{(\infty , \beta_R = 1)} (E)$ (solid), $\vel\mapsto\frac{1}{2} \Rcal_{\vel}^{(\beta=1)} (E)$ (dashed)}
		\label{fig:sub2-ness-moving-infty}
	\end{subfigure}
	\caption{\textbf{(a)} Surface plot of $(\vel,E)\mapsto \Rcal_{N,\vel}^{(\infty , \beta_R)} (E)$ for $\beta_R = 1$ and $\vel\in(-1,1)$, $E>0$. The dashed surface curve at $\vel=0$ represents the transition rate of a stationary detector and is given by $\frac{1}{2} \Rcal_{0}^{(\beta=1)}$. \quad \textbf{(b)} Plot of the same transition rate as a function of $\vel$ for different values of $E>0$ (solid lines). The corresponding rescaled transition rate $\frac{1}{2} \Rcal_{\vel}^{(\beta=1)}$ for the $(\beta=1)$-KMS heat bath alone (Eq.\ \eqref{eq:transition-rate-kms-moving}) is shown by dashed lines.}
	\label{fig:rate-ness-moving-infty}
\end{figure}

Let us now consider the limit case $\beta_L \to \infty$ for which the left reservoir approaches the ground state (acting as a perfect ``heat sink''). 
(The situation with $\beta_R \to \infty$ and finite $\beta_L$ is effectively obtained from this by inverting the direction of motion, $\vel\mapsto -\vel$.) 
Since $\lim_{\beta\to\infty} R(\beta,\vel,E) = 0$ for $E>0$ and $\vel\in(-1,1)\setminus\{0\}$ (see Eq.\ \eqref{eq:R-function}), the transition rate in this limit is given by $R(\beta_R , \vel , E)$, i.e.\
\begin{gather*}
	\Rcal_{N,\vel}^{(\infty , \beta_R)} (E) := \lim\limits_{\beta_L \to \infty} \Rcal_{N,\vel}^{(\beta_L , \beta_R)} (E) = \frac{\sqrt{1-\vel^2}}{4\pi\vel\beta_R} \ln\left(\frac{1-\e^{-\beta_R E \sqrt{1-\vel^2}}}{1-\e^{-\beta_R E \sqrt{\frac{1-\vel}{1+\vel}}}} \right) \, , \quad \beta_R , E>0 \, .
\end{gather*}
The plots in Figure \ref{fig:rate-ness-moving-infty} visualize this transition rate (for $\beta_R = 1$). 
For every fixed $\vel$ and $\beta_R$, the function $\Rcal_{N,\vel}^{(\infty , \beta_R)}$ is given by $\sqrt{1-\vel^2}(4\pi\vel\beta_R)^{-1} \ln(1+\vel)$ in the limit $E\to 0$ (see Eq.\ \eqref{eq:transition-rate-ness-vel-small-gap}), and is decreasing in $E>0$. 

We can contrast this transition rate with the scenario \textsf{(KMS-$\vel$)} in which the initial Minkowski vacuum state in the left half-space (and thus the heat flow to $x^1 = -\infty$) is absent and the detector interacts with a $\beta_R$-KMS state. 
To this end, we consider in Figure \ref{fig:sub2-ness-moving-infty} the rescaled transition rate $\frac{1}{2} \Rcal_{\vel}^{(\beta_R)}$ (Eq.\ \eqref{eq:transition-rate-kms-moving}), where the factor $\frac{1}{2}$ accommodates for the limit $\vel\to 0$ in which $\Rcal_{N,0}^{(\infty , \beta_R)} (E) = \frac{1}{2} \Rcal_{0}^{(\beta_R)} (E)$ for $E>0$ by Eq.\ \eqref{eq:transition-rate-ness-stat}. 
The relation between Figure \ref{fig:rate-ness-moving-infty} and Figure \ref{fig:sub1-kms-moving} shows the asymmetry in $\vel$ discussed above. 

For a comparison with the case \textsf{(NESS-$0$)} we examine the quotient $\Rcal_{N,\vel}^{(\infty , \beta_R)} (E) / \Rcal_{N,0}^{(\infty , \beta_R)} (E)$ as a function of $\vel\in(-1,1)$ parametrized by $\beta_R E > 0$ (see Figure \ref{fig:rate-ness-moving-infty2}). 
Similar to the quotient $\Rcal_{\vel}^{(\beta)} (E) / \Rcal_{0}^{(\beta)} (E)$ in the case \textsf{(KMS-$\vel$)}, which has been discussed in \cite{Costa2004}, we can find the ``critical value'' of $\beta_R E > 0$ that determines the change of slope in $\vel=0$ that occurs when passing from ``smaller'' to ``larger'' values of $\beta_R E$. 
\begin{corollary}
	\label{cor:ness-quotient}
	The quotient of the transition rates for \textup{\textsf{(NESS-$\vel$)}} and \textup{\textsf{(NESS-$0$)}} (with $E>0$ and $\beta_L \to \infty$) is given by
	\begin{gather*}
		\frac{\Rcal_{N,\vel}^{(\infty , \beta_R)} (E)}{\Rcal_{N,0}^{(\infty , \beta_R)} (E)} = \frac{\sqrt{1-\vel^2} (\e^{\beta_R E} - 1)}{\vel\beta_R E} \ln\left(\frac{1-\e^{-\beta_R E \sqrt{1-\vel^2}}}{1-\e^{-\beta_R E \sqrt{\frac{1-\vel}{1+\vel}}}} \right) \, , \quad \beta_R , E>0 \, .
	\end{gather*}
	In the limit $\beta_R E \to 0$ this equals $\sqrt{1-\vel^2} \, \vel^{-1} \ln(1+\vel)$ for every fixed $\vel$. 
	The quotient $\Rcal_{N,\vel}^{(\infty , \beta_R)} (E) / \Rcal_{N,0}^{(\infty , \beta_R)} (E)$ has vanishing first derivative in $\vel=0$ if and only if $\beta_R E$ takes the value
	\begin{gather*}
		c_N := 2 + \mathrm{W}(-2\e^{-2}) \approx 1.59362 \, ,
	\end{gather*}
	where $\mathrm{W}$ denotes the principal branch of the Lambert W-function \cite[4.13]{NIST:DLMF}.
\end{corollary}
\begin{proof}
	The expression for the quotient is a consequence of the previous results, and the limit $\beta_R E \to 0$ can be deduced from Eq.\ \eqref{eq:transition-rate-ness-vel-small-gap} and $\lim_{\beta_R E \to 0} \beta_R \Rcal_{N,0}^{(\infty , \beta_R)} (E) = \frac{1}{4\pi}$ by Eq.\ \eqref{eq:transition-rate-ness-stat}. 
	A calculation of the first two terms in the Taylor expansion around $\vel=0$ shows that
	\begin{gather*}
		\frac{\Rcal_{N,\vel}^{(\infty , \beta_R)} (E)}{\Rcal_{N,0}^{(\infty , \beta_R)} (E)} = 1 + \frac{\e^{\beta_R E} (\beta_R E - 2) + 2}{2(\e^{\beta_R E} - 1)} \, \vel + \mathcal{O}(\vel^2) \, .
	\end{gather*}
	From this we see that the first derivative of the quotient in $\vel=0$ vanishes if and only if $\e^{\beta_R E} (\beta_R E - 2) + 2 = 0$, which has the unique positive solution $\beta_R E = c_N$ as stated (see \cite[4.13]{NIST:DLMF}). 
\end{proof}
The plots in Figure \ref{fig:rate-ness-moving-infty2} again show the dependence on the sign of $\vel$. 
At $\beta_R E = c_N$ the quotient is smaller than $1$ for all $\vel\in(-1,1)\setminus\{0\}$. 
The critical value of $\beta E$ for the analogous quotient $\Rcal_{\vel}^{(\beta)} (E) / \Rcal_{0}^{(\beta)} (E)$ in the case \textsf{(KMS-$\vel$)} is determined by the vanishing second derivative in $\vel=0$ (note that the expansion of $\Rcal_{\vel}^{(\beta)}$ only depends on even powers of $\vel$ as a manifestation of symmetry \cite{Costa2004}) and is given by the solution to $\beta E \coth(\beta E/2)=3$, which is approximately $2.57568$ \cite{Costa2004}. 
Hence the critical value $c_N \approx 1.59362$ in the NESS case is about $38\%$ less the value in the KMS case. 
This indicates that, in comparison with the interaction with a KMS state, the presence of the vacuum state reservoir in the left half-space (and the NESS it entails at late times) has a non-trivial influence on the velocity-dependence of the transition rate. 

\begin{figure}[!t]
	\centering
	\includegraphics[width=.6\textwidth]{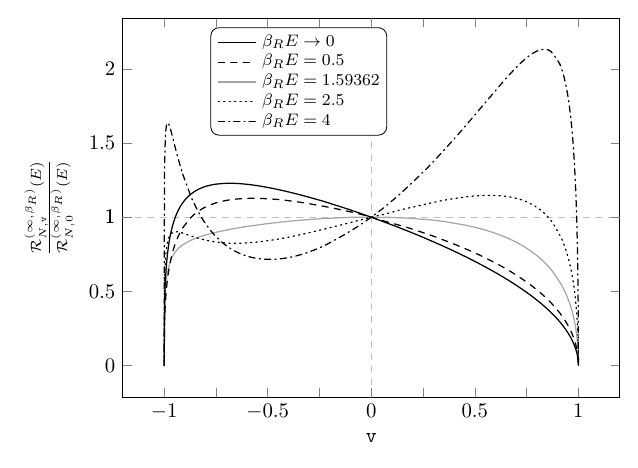}
	\caption{Plots of the quotient of $\Rcal_{N,\vel}^{(\infty , \beta_R)} (E)$ and $\Rcal_{N,0}^{(\infty , \beta_R)} (E) = \frac{1}{2} \Rcal_{0}^{(\beta_R)} (E)$ (for $E>0$) as a function of $\vel$ in the limit $\beta_R E \to 0$ (solid) and for $\beta_R E = 0.5$ (dashed), $2.5$ (dotted), $4$ (dash-dotted), and the critical value $2 + \mathrm{W}(-2\e^{-2}) \approx 1.59362$ (solid gray).}
	\label{fig:rate-ness-moving-infty2}
\end{figure}

\section{Detailed balance effective temperature}
\label{sec:detailed-balance}

The detailed balance condition \cite{Takagi1986,Waiting-for-Unruh} for the ratio of excitation and de-excitation probabilities of the detector allows to assign an effective temperature to the response of a detector \cite{Good-JA-Moustos-Temirkhan2020,Biermann-et-al2020,Bunney-Louko2023,Bunney-Parry-Perche-Louko2024,Parry-Louko2025} (see also \cite{Bell-Leinaas1983,Unruh1998}). 
Following these references, we define the effective temperature perceived by a monopole detector and discuss some of its features for inertially moving detectors coupled to the NESS $\omega_N$. 
\begin{definition}[Detailed balance effective temperature]
	\label{def:detailed-balance}
	For a detector with energy gap $E\in\R\setminus\{0\}$ and transition rate $\Rcal_\omega$ under the coupling to a field state $\omega$ (Eq.\ \eqref{eq:transition-rate}), the \emph{detailed balance effective temperature} is (provided it exists)
	\begin{gather*}
		\Teff(E) := \frac{E}{\ln\left(\frac{\Rcal_\omega (-E)}{\Rcal_\omega (E)}\right)} \, ,
	\end{gather*}
	i.e.\ $\Teff(E)$ is the quantity such that $\Rcal_\omega (E) = \e^{-E/\Teff(E)} \Rcal_\omega (-E)$.
\end{definition}
Notice that $\Teff(-E)=\Teff(E)$ for all $E\in\R\setminus\{0\}$. 
The detailed balance condition $\Rcal_\omega (E) = \e^{-\beta E} \Rcal_\omega (-E)$ for a constant $\beta>0$ is related to a form of the $\beta$-KMS condition for the two-point function pulled back along the detector's stationary worldline \cite{Waiting-for-Unruh}. 
If $\omega$ is a $\beta$-KMS state relative to the detector's rest frame (case \textsf{(KMS-$0$)}) then $\Teff(E)$ is independent of $E$ and coincides with the temperature $\beta^{-1}$, as can be seen by inserting $\Rcal_{0}^{(\beta)}$ from Eq.\ \eqref{eq:transition-rate-kms-rest}. 
In general, however, the detailed balance effective temperature $\Teff(E)$ will depend on the energy gap of the detector and thus does not represent an equilibrium temperature in the Gibbs sense (cf.\ \cite{Parry-Louko2025} for comments in that direction). 
While the significance and interpretation of $\Teff(E)$ as a thermodynamic quantity might be up to debate, it nevertheless can characterize the asymptotic readings of the detector for given $E$ under the interaction with the quantum field (see \cite{Good-JA-Moustos-Temirkhan2020,Biermann-et-al2020,Bunney-Louko2023,Bunney-Parry-Perche-Louko2024,Parry-Louko2025} for comments and some recent applications, and further references therein). 
In \cite{JuarezAubry-Moustos2019} it has been shown that for detectors moving along stationary trajectories (orbits of timelike Killing vector fields) and interacting with states for which the detailed balance condition $\Rcal_\omega (E) = \e^{-E/\Teff(E)} \Rcal_\omega (-E)$ holds, the reduced density matrix of the detector at asymptotically late times under Born-Markov approximation has a Gibbs form corresponding to the (possibly energy-dependent) effective temperature $\Teff(E)$.\medskip 

Let us denote the effective temperatures for the cases \textsf{(KMS-$\vel$)} and \textsf{(NESS-$\vel$)} by $\Teff_{\vel}^{(\beta)}$ and $\Teff_{N,\vel}^{(\beta_L , \beta_R)}$ for $\vel\in(-1,1)$, respectively. 
As noted above, $\Teff_{0}^{(\beta)} \equiv \beta^{-1}$ for every $\beta>0$. 
For $\vel\neq 0$, Eqs.\ \eqref{eq:transition-rate-kms-moving} \& \eqref{eq:transition-rate-kms-moving-small-gap} imply that 
\begin{gather}
	\label{eq:Tvbeta}
	\Teff_{\vel}^{(\beta)} (E) = E \left[ \ln\left( \frac{\ln\left(\frac{1-\e^{\beta \dop_+ E}}{1-\e^{\beta \dop_- E}} \right)}{\ln\left(\frac{1-\e^{-\beta \dop_+ E}}{1-\e^{-\beta \dop_- E}} \right)} \right) \right]^{-1} = \frac{\sqrt{1-\vel^2}}{2\beta\vel} \ln\left(\frac{1+\vel}{1-\vel}\right) + \mathcal{O}(E^2) \, .
\end{gather}
The lowest order of Eq.\ \eqref{eq:Tvbeta} coincides with \cite[Eq.\ (8)]{Landsberg-Matsas1996}, the solid angle average of a ``directional temperature''; see Appendix \ref{appendix:eff-temp} for some details. 
For small $|\vel|$ and $|E|$ the detailed balance temperature $\Teff_{\vel}^{(\beta)}$ is approximately 
\begin{gather}
	\label{eq:eff-temp-kms-slow}
	\Teff_{\vel}^{(\beta)} (E) \simeq \beta^{-1} \left(1-\frac{\vel^2}{6} \right) \, , \quad |\vel|\ll 1 \, , \, E\to 0 \, ,
\end{gather}
which is strictly smaller than the rest frame temperature $\beta^{-1}$ and has been given in \cite[Eq.\ (9)]{Costa-Matsas1995} as the effective temperature defined by a moving observer at small velocities and in the infrared regime of the radiation spectrum. 
We note that other spectral regimes and velocities relate to differing results, some of which may be higher than the rest frame temperature \cite{Costa-Matsas1995}. 
This is closely related to the argument in \cite{Costa-Matsas1995,Landsberg-Matsas1996,Landsberg-Matsas2004} against the existence of a relativistic transformation law for temperature (see also \cite{Farias-Pinto-Moya} and \cite[Sec.\ 7.7]{CasasVazquez-Jou2003} for a review of the various approaches and positions on this matter).\medskip 

We now turn to the NESS case. 
The following proposition summarizes the main properties of the detailed balance effective temperatures expanded in $E$, if one of the heat baths is replaced by the vacuum state ($\beta_L \to \infty$), and in the limit of small velocities. 
The results follow from Definition \ref{def:detailed-balance}, Eqs.\ \eqref{eq:transition-rate-ness-stat} \& \eqref{eq:transition-rate-ness-vel}, and expansions in small $|E|$ (see Eq.\ \eqref{eq:transition-rate-ness-vel-small-gap}). 
\begin{proposition}[Effective temperatures for inertial detectors coupled to NESS]
	\label{prop:eff-temp}
	Let $\Teff_{N,\vel}^{(\beta_L , \beta_R)}$ be the detailed balance effective temperature for \textup{\textsf{(NESS-$\vel$)}}, $\vel\in(-1,1)$, $\beta_L , \beta_R > 0$. 
	\begin{itemize}
		\item For \textup{\textsf{(NESS-$0$)}}, 
		\begin{gather}
			\label{eq:eff-temp-ness-0}
			\Teff_{N,0}^{(\beta_L , \beta_R)} (E) = E \left[\ln\left(\frac{2\e^{(\beta_L + \beta_R)E} - \e^{\beta_L E} - \e^{\beta_R E}}{\e^{\beta_L E} + \e^{\beta_R E} - 2}\right)\right]^{-1} = \frac{1}{2} (\beta_L^{-1} + \beta_R^{-1}) + \mathcal{O}(E^2) \, ,
		\end{gather}
		with $\Teff_{N,0}^{(\beta , \beta)} = \Teff_{0}^{(\beta)} \equiv \beta^{-1}$ for all $\beta>0$, and
		\begin{gather*}
			\Teff_{N,0}^{(\infty , \beta_R)} (E) := \lim\limits_{\beta_L \to \infty} \Teff_{N,0}^{(\beta_L , \beta_R)} (E) = \frac{|E|}{\ln(2\e^{\beta_R |E|} -1)} = \frac{1}{2\beta_R} + \frac{|E|}{4} + \mathcal{O}(E^2) \, .
		\end{gather*}
		\item For \textup{\textsf{(NESS-$\vel$)}} with $\vel\in(-1,1)\setminus\{0\}$, 
		\begin{gather}
			\label{eq:eff-temp-ness-v}
			\Teff_{N,\vel}^{(\beta_L , \beta_R)} (E) = \frac{\sqrt{1-\vel^2}}{2\vel} \left( \frac{1}{\beta_R} \ln(1+\vel) - \frac{1}{\beta_L} \ln(1-\vel) \right) + \mathcal{O}(E^2)
		\end{gather}
		and
		\begin{gather}
			\label{eq:eff-temp-ness-v-infty}
			\Teff_{N,\vel}^{(\infty , \beta_R)} (E) := \lim\limits_{\beta_L \to \infty} \Teff_{N,\vel}^{(\beta_L , \beta_R)} (E) = \frac{|E|}{\ln\left(1+\frac{|E|}{2\pi R(\beta_R , \vel , |E|)}\right)}
		\end{gather}
		for the function $R$ defined in Eq.\ \eqref{eq:R-function}. 
		For slowly moving detectors only probing the infrared (low energy) regime, 
		\begin{gather}
			\label{eq:Teff-N-v}
			\Teff_{N,\vel}^{(\beta_L , \beta_R)} (E) \simeq \frac{1}{2} (\beta_L^{-1} + \beta_R^{-1}) + \frac{\beta_R - \beta_L}{4\beta_L \beta_R} \vel - \frac{\beta_L + \beta_R}{12\beta_L \beta_R} \vel^2 \, , \quad |\vel|\ll 1 \, , \, E\to 0 \, .
		\end{gather}
	\end{itemize}
\end{proposition}

The functional properties of $\Teff_{N,0}^{(\beta_L , \beta_R)}$ are illustrated in Figures \ref{fig:temp-ness0-3} \& \ref{fig:temp-ness0}. 
Once again, we see that the response of a stationary detector in the NESS does not correspond to a thermal equilibrium in the Gibbs sense (cf.\ the paragraph below Proposition \ref{prop:ness-stat}). 
Note that for $\beta_L > \beta_R$ and $E>0$, 
\begin{gather*}
	\lim\limits_{E\to \infty} \frac{1}{E} \ln\left(\frac{2\e^{(\beta_L + \beta_R)E} - \e^{\beta_L E} - \e^{\beta_R E}}{\e^{\beta_L E} + \e^{\beta_R E} - 2}\right) = \lim\limits_{E\to \infty} \frac{\ln(2\e^{\beta_R E} - 1)}{E} = \beta_R
\end{gather*}
by continuity, which means that for large $|E|$ the hotter heat bath dominates in the effective temperature: 
\begin{gather*}
	\lim\limits_{|E|\to \infty} \Teff_{N,0}^{(\beta_L , \beta_R)} (E) = \frac{1}{\min(\beta_L , \beta_R)}
\end{gather*}
Hence for $|E|\to\infty$ the effective temperature $\Teff_{N,0}^{(\infty , \beta_R)} (E)$ approaches $\beta_R^{-1}$, while for $E\to 0$ it approaches $(2\beta_R)^{-1}$.\medskip 

\begin{figure}[!t]
	\centering
	\includegraphics[width=\textwidth]{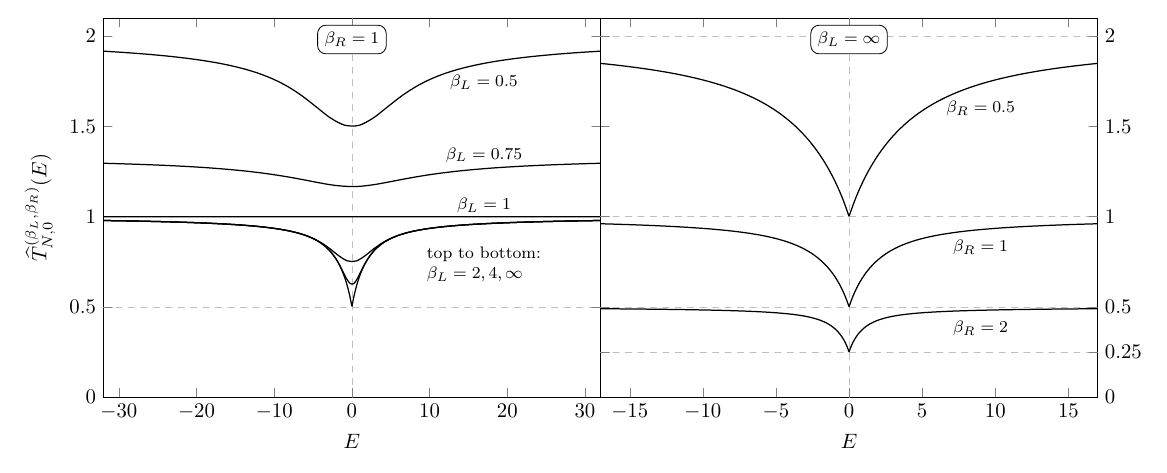}
	\caption{Plot of the effective temperature $\Teff_{N,0}^{(\beta_L , \beta_R)}$, for $\beta_R = 1$ and different values of $\beta_L$ (left panel), and for $\beta_L = \infty$ and different values of $\beta_R$ (right panel). This illustrates the limits $\lim_{E\to 0} \Teff_{N,0}^{(\beta_L , \beta_R)} (E) = (\beta_L^{-1} + \beta_R^{-1})/2$ and $\lim_{|E|\to \infty} \Teff_{N,0}^{(\beta_L , \beta_R)} (E) = (\min(\beta_L , \beta_R))^{-1}$, with the thermal case $\Teff_{N,0}^{(\beta , \beta)} = \Teff_{0}^{(\beta)} \equiv \beta^{-1}$ for $\beta > 0$.}
	\label{fig:temp-ness0-3}
\end{figure}
\begin{figure}[!t]
	\centering
	\includegraphics[width=.6\textwidth]{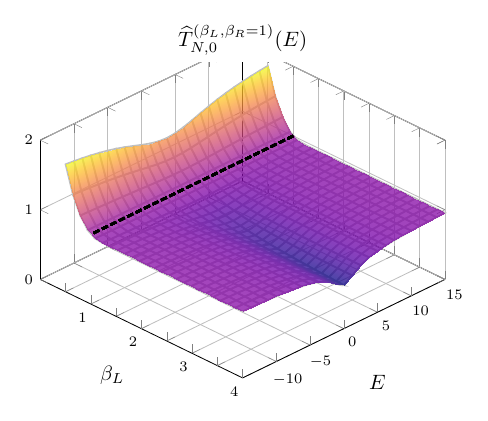}
	\caption{Surface plot of $(\beta_L , E)\mapsto\Teff_{N,0}^{(\beta_L , \beta_R)} (E)$ for $\beta_R = 1$ and $\beta_L > 0.5$, corresponding to the left panel in Figure \ref{fig:temp-ness0-3}. The dashed surface curve at $\beta_L = \beta_R = 1$ represents the constant KMS heat bath temperature $\Teff_{0}^{(\beta=1)} = 1$ (as a consequence of $\Rcal_{N,0}^{(\beta , \beta)} = \Rcal_{0}^{(\beta)}$).}
	\label{fig:temp-ness0}
\end{figure}

\begin{figure}[!t]
	\centering
	\includegraphics[width=.6\textwidth]{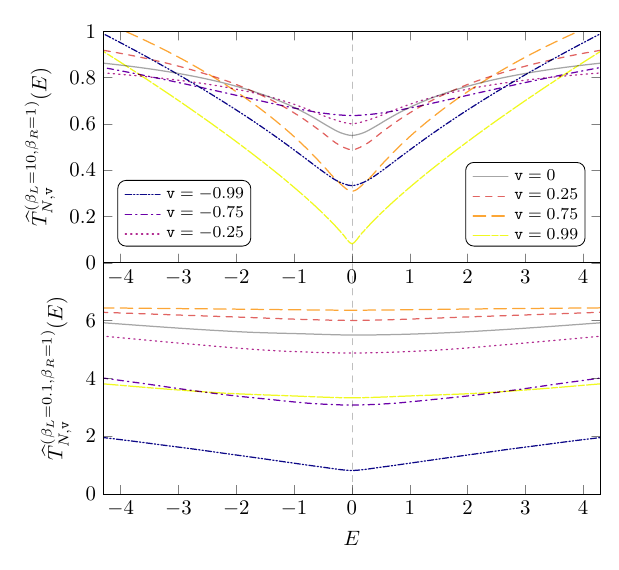}
	\caption{Plot of the effective temperature $\Teff_{N,\vel}^{(\beta_L , \beta_R)}$ for $\beta_R = 1$, and $\beta_L = 10$ (top panel) and $\beta_L = 0.1$ (bottom panel), for different values of $\vel$ as indicated in the legend. The zero velocity limit $\Teff_{N,0}^{(\beta_L , \beta_R)}$ is shown by the solid gray curves.}
	\label{fig:temp-nessv-2d}
\end{figure}

\begin{figure}[!t]
	\centering
	\begin{subfigure}{.5\textwidth}
		\centering
		\includegraphics[width=\textwidth]{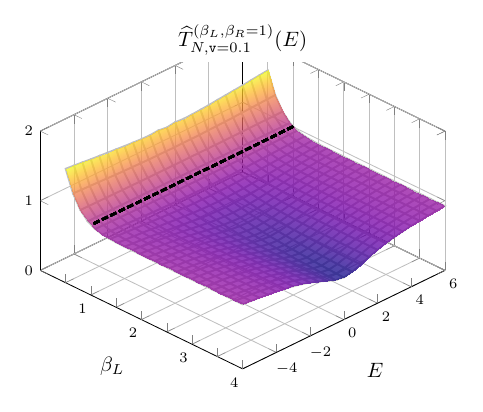}
		\caption{$(\beta_L , E) \mapsto \Teff_{N,\vel=0.1}^{(\beta_L , \beta_R = 1)} (E)$ ($\beta_L > 0.5$)}
		\label{fig:temp-nessv-1}
	\end{subfigure}%
	\begin{subfigure}{.5\textwidth}
		\centering
		\includegraphics[width=\textwidth]{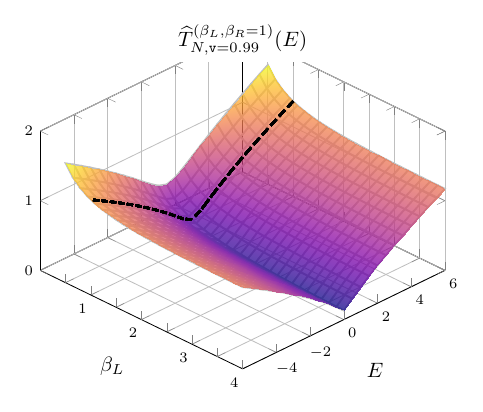}
		\caption{$(\beta_L , E) \mapsto \Teff_{N,\vel=0.99}^{(\beta_L , \beta_R = 1)} (E)$ ($\beta_L > 0.5$)}
		\label{fig:temp-nessv-2}
	\end{subfigure}
	\caption{Surface plots of the effective temperature $(\beta_L , E)\mapsto\Teff_{N,\vel}^{(\beta_L , \beta_R)} (E)$ for $\beta_R = 1$, $\beta_L > 0.5$, and $\vel=0.1$ (left), $0.99$ (right). The dashed surface curves at $\beta_L = \beta_R = 1$ correspond to $E\mapsto \Teff_{\vel}^{(\beta=1)} (E)$ given by Eq.\ \eqref{eq:Tvbeta} (as a consequence of $\Rcal_{N,\vel}^{(\beta , \beta)} = \Rcal_{\vel}^{(\beta)}$).}
	\label{fig:temp-nessv}
\end{figure}

Figure \ref{fig:temp-nessv-2d} shows the dependence of $\Teff_{N,\vel}^{(\beta_L , \beta_R)}$ on $\vel$ and the KMS parameters of the semi-infinite heat baths in a neighborhood of $E=0$. 
For concreteness, we have chosen $\beta_R = 1$ and $\beta_L \in \{0.1,10\}$. 
(The properties of the limit $\beta_L \to \infty$, given by Eq.\ \eqref{eq:eff-temp-ness-v-infty}, will not be studied here.) 
One observes that for fixed KMS parameters the effective temperature may be higher or lower than $\Teff_{N,0}^{(\beta_L , \beta_R)}$ (the solid gray curves in the figure) for certain velocities, and it is not an even function of $\vel$. 
Moreover, the relationship between the effective temperatures at different velocities changes depending on the KMS parameters, more specifically on whether $\beta_L$ is larger or smaller than $\beta_R$. 
We see that the effective temperature $\Teff_{N,\vel}^{(\beta_L , \beta_R)}$ near $E=0$ exceeds $\Teff_{N,0}^{(\beta_L , \beta_R)}$ when the detector moves in the direction of the heat flow, except for $|\vel|=0.99$ near the speed of light. 
(Recall from Eq.\ \eqref{eq:inertial-worldline} that the detector moves from left to right towards $x^1 = +\infty$ for $\vel>0$, and reverse for $\vel<0$.) 
Figure \ref{fig:temp-nessv} illustrates the dependence of $\Teff_{N,\vel}^{(\beta_L , \beta_R)}$ on $\beta_L$ in a surface plot for $\beta_R = 1$ and arbitrarily chosen velocities $\vel\in\{0.1,0.99\}$. 

Compared to \textsf{(KMS-$\vel$)} (Eq.\ \eqref{eq:Tvbeta}), the behavior of $\Teff_{N,\vel}^{(\beta_L , \beta_R)} (E)$ in $\vel$ for small $|E|$ (see Eq.\ \eqref{eq:eff-temp-ness-v}) is similar to the transition rates shown in the left plot of Figure \ref{fig:rate-ness-moving}: 
Unless $\beta_L = \beta_R$, there is an asymmetry in $\vel$ and a maximum value of $\Teff_{N,\vel}^{(\beta_L , \beta_R)} (E)$ for some $|\vel|>0$ depending on $\sgn(\beta_L - \beta_R)$. 
For slowly moving detectors only probing the infrared regime (i.e.\ $|\vel|\ll 1$, $E\to 0$), we notice that $\Teff_{\vel}^{(\beta)}$ given in Eq.\ \eqref{eq:eff-temp-kms-slow} is strictly smaller than the corresponding rest frame temperature $\beta^{-1}$. 
By contrast, the right-hand side of Eq.\ \eqref{eq:Teff-N-v} is larger than $\lim_{E\to 0} \Teff_{N,0}^{(\beta_L , \beta_R)} (E) = \frac{1}{2} (\beta_L^{-1} + \beta_R^{-1})$ (Eq.\ \eqref{eq:eff-temp-ness-0}) for positive (negative) $\vel$ with $|\vel|\ll 1$ when $\beta_L < \beta_R$ ($\beta_L > \beta_R$).

\section{The comoving frame}
\label{sec:case-NESS-comoving}

In this section we introduce the comoving frame of the NESS as a proposal for a reference frame that comes closest to the rest frame of a thermal equilibrium state and thus could be a suitable frame for the comparison of detector responses in future investigations. 

For $\beta_R > \beta_L > 0$ there is a heat flow from $x^1 = -\infty$ to $x^1 = \infty$ in the NESS described by $\omega_N$ (and in the reverse direction for $0<\beta_R < \beta_L$) \cite{DLSB,Hack-Verch}. 
The idea is to consider the special case of \ref{item:NESS} in which the detector couples to the NESS in a reference frame where the heat flow is ``at rest'', i.e.\ the detector is comoving with the heat flow at a certain constant velocity that depends on $\beta_L , \beta_R$. 
This frame is characterized by the property that the expected stress-energy tensor of the field is diagonalized, in analogy to the (local) rest frame of a perfect fluid (which is fully described by a non-zero energy density and isotropic pressure \cite{Schutz-GR}). 
The heat flow is thereby rendered as a velocity effect that is compensated kinematically by boosting to the comoving rest frame of the flow. 
A stationary detector in this frame is at rest relative to the ``medium'' of the NESS, similar to the thermal case \hyperref[item:KMS]{\textsf{(KMS-$0$)}}. 
We argue that the detector behavior should be studied relative to this frame, in the spirit of the (widely supported) guiding principle that thermal properties of systems ought to be determined by distinguished quantities defined relative to their rest frame; see, e.g., \cite{Cavalleri-Salgarelli,Callen-Horwitz,Costa-Matsas1995,Landsberg-Matsas1996,Landsberg-Matsas2004,Sewell2008,Sewell-rep2009,Sewell2010,PV-disj}. 
This also constitutes a fundamental premise for the definition of local thermal equilibrium states of quantum fields \cite{Buchholz-Ojima-Roos}. \medskip

The (renormalized) stress-energy tensor of the massless scalar field $\phi$ is given by (see, e.g., \cite{Birrell-Davies,Wald-book,Khavkine-Moretti})
\begin{gather}
	\label{eq:set}
	T_{\mu\nu} = \normord{\partial_\mu \phi \, \partial_\nu \phi - \frac{1}{2} \eta_{\mu\nu} \partial_\rho \phi \, \partial^\rho \phi} \, ,
\end{gather}
where $\eta$ is the Minkowski metric tensor, the colons denote normal ordering with respect to the vacuum state, and Einstein's summation convention is implied. 
Using Eq.\ \eqref{eq:ness-two-point-reg} we show in Lemma \ref{lem:set-ness-components} (rederiving results from \cite{DLSB}) that the expected stress-energy tensor in the NESS $\omega_N$ (constructed in the inertial frame $I$) has block diagonal form, namely
\begin{gather*}
	(\Tcal_{\mu\nu}) := (\omega_N (T_{\mu\nu})) = \begin{pmatrix} \widetilde{\Tcal} & \nullmatrix_2 \\ \nullmatrix_2 &  \diag (\Tcal_{\perp},\Tcal_{\perp}) \end{pmatrix} \, ,
\end{gather*}
where $\nullmatrix_2$ is the $2\times 2$ zero matrix, and
\begin{gather*}
	\widetilde{\Tcal} := \begin{pmatrix} \Tcal_{00} & \Tcal_{01} \\ \Tcal_{01} & \Tcal_{11} \end{pmatrix} \, , \quad \Tcal_{\perp} := \Tcal_{22} = \Tcal_{33}
\end{gather*}
for non-zero expected energy current density $\Tcal_{01}$ \cite{DLSB}. 

\begin{proposition}[Comoving frame]
	\label{prop:comoving}
	Let $\beta_L , \beta_R >0$, $\beta_L \neq \beta_R$. 
	There exists a Lorentz boost in the $x^1$-direction that diagonalizes the expected stress-energy tensor in the NESS $\omega_N$. 
	The corresponding velocity $\vel_N (\beta_L , \beta_R)$ of the boost is given by
	\begin{gather*}
		\vel_N (\beta_L , \beta_R) = \kappa(\beta_L , \beta_R) - \sgn(\beta_R - \beta_L) \sqrt{\kappa(\beta_L , \beta_R)^2 -1} \, , \quad \kappa(\beta_L , \beta_R) := \frac{4}{3} \, \frac{\beta_R^4 + \beta_L^4}{\beta_R^4 - \beta_L^4} \, ,
	\end{gather*}
	which can be written as 
	\begin{gather*}
		\vel_N (\beta_L , \beta_R) = \frac{4(\beta_R^4 + \beta_L^4)-\sqrt{(\beta_R^4 + 7\beta_L^4)(\beta_L^4 + 7\beta_R^4)}}{3(\beta_R^4 - \beta_L^4)}
	\end{gather*}
	and satisfies the bound $|\vel_N (\beta_L , \beta_R)| < \frac{4-\sqrt{7}}{3} \approx 0.4514$. 
	The diagonalized expected stress-energy tensor is not isotropic, i.e.\ its $(11)$ and $(22)$ components are distinct.
\end{proposition}
The reference frame moving along the $x^1$-axis with velocity $\vel_N (\beta_L , \beta_R)$ relative to $I$ is called the \emph{comoving frame} of the NESS. 
\begin{proof}
	Let $\Uplambda=({\Uplambda^{\mu}}_{\nu})$ be a Lorentz boost in $x^1$-direction given by
	\begin{gather}
		\label{eq:boost-x1-direction}
		\Uplambda = \begin{pmatrix} \widetilde{\Uplambda} & \nullmatrix_2 \\ \nullmatrix_2 & \mathds{1}_2 \end{pmatrix} \, , \quad \widetilde{\Uplambda} = \begin{pmatrix} \gam & -\gam\vel \\ -\gam\vel & \gam \end{pmatrix} \, ,
	\end{gather}
	where $\mathds{1}_2$ is the $2\times 2$ unit matrix, and $\vel\in(-1,1)$. 
	It holds ${\Uplambda_{\mu}}^{\rho} {\Uplambda_{\nu}}^{\sigma} \, \Tcal_{\rho\sigma} = 0$ for all $\mu\neq\nu$ if and only if the off-diagonal component of $\widetilde{\Uplambda}^{-1} \widetilde{\Tcal} \widetilde{\Uplambda}^{-1}$ vanishes. 
	A straightforward calculation shows that this happens for $\vel$ equal to
	\begin{gather*}
		\vel_\pm = -\frac{\Tcal_{00} + \Tcal_{11}}{2\Tcal_{01}} \pm \sqrt{\left(\frac{\Tcal_{00} + \Tcal_{11}}{2\Tcal_{01}}\right)^2 -1}
	\end{gather*}
	under the condition $\vel_\pm \in (-1,1)$. 
	From Lemma \ref{lem:set-ness-components} it follows that 
	\begin{gather*}
		-\frac{\Tcal_{00}+\Tcal_{11}}{2\Tcal_{01}} = \frac{4}{3} \, \frac{\beta_R^4 + \beta_L^4}{\beta_R^4 - \beta_L^4} =: \kappa(\beta_L , \beta_R) \, .
	\end{gather*}
	We have $\kappa(\beta_L , \beta_R)>\frac{4}{3}$ (respectively, $\kappa(\beta_L , \beta_R)<-\frac{4}{3}$) if and only if $\beta_R > \beta_L > 0$ (respectively, $0<\beta_R < \beta_L$), in which case the unique velocity of the Lorentz boost diagonalizing the expected stress-energy tensor is $\vel_- \in (0,\frac{4-\sqrt{7}}{3})$ (respectively, $\vel_+ \in (-\frac{4-\sqrt{7}}{3},0)$). 
	Combining these two cases yields the velocity $\vel_N (\beta_L , \beta_R)$ as stated.\medskip
	
	For a Lorentz boost $\Uplambda$ in $x^1$-direction as in Eq.\ \eqref{eq:boost-x1-direction}, the $(11)$ component of $\Uplambda^{-1} (\mathcal{T}_{\mu\nu}) \Uplambda^{-1}$ is given by $\gam^2 (\Tcal_{11} + 2\Tcal_{01} \vel + \Tcal_{00} \vel^2)$, whereas the $(22)$ component equals $\frac{1}{3} \Tcal_{00} = \Tcal_{11}$ by Lemma \ref{lem:set-ness-components}. 
	Hence the $(11)$ and $(22)$ components coincide either when $\vel=0$ or
	\begin{gather*}
		\vel = -\frac{3\Tcal_{01}}{2\Tcal_{00}} = \frac{3}{4} \, \frac{\beta_R^4 - \beta_L^4}{\beta_R^4 + \beta_L^4} \equiv \frac{1}{\kappa(\beta_L , \beta_R)} \, .
	\end{gather*}
	Since $0<|\vel_N (\beta_L , \beta_R)| < |\kappa(\beta_L , \beta_R)^{-1}|$ for all $\beta_L , \beta_R > 0$, $\beta_L \neq \beta_R$, the $(11)$ and $(22)$ components of the expected stress-energy tensor in the comoving frame do not coincide. 
\end{proof}

Some remarks on the velocity of the comoving frame are in order. 
It is antisymmetric in the initial KMS parameters $\beta_L , \beta_R$ of the NESS, i.e.\ $\vel_N (\beta_L , \beta_R)=-\vel_N (\beta_R , \beta_L)$. 
This reflects the physical picture that the direction of the heat flow is reversed when the thermal reservoirs are swapped. 
The bound $|\vel_N (\beta_L , \beta_R) | < \frac{4-\sqrt{7}}{3}$ is an interesting property that is a consequence of the expectation value of the stress-energy tensor in four spacetime dimensions (Lemma \ref{lem:set-ness-components}). 
If one of the semi-infinite heat baths is prepared in the vacuum state, i.e.\ $\beta_L > 0$ and $\beta_R \to \infty$ or vice versa, we have $|\kappa(\beta_L , \beta_R)|\to\frac{4}{3}$ and thus $|\vel_N (\beta_L , \beta_R)|\to\frac{4-\sqrt{7}}{3} \approx 0.4514$. 
Similarly, $|\vel_N (\beta_L , \beta_R)|\to\frac{4-\sqrt{7}}{3}$ when one of the reservoirs approaches infinitely high temperature, $\beta_{L/R} \to 0$. 
Hence the heat flow of the NESS always travels subluminally (see Figure \ref{fig:sub1-ness-comoving}).

\begin{figure}[!t]
	\centering
	\begin{subfigure}{.5\textwidth}
		\centering
		\includegraphics[width=\textwidth]{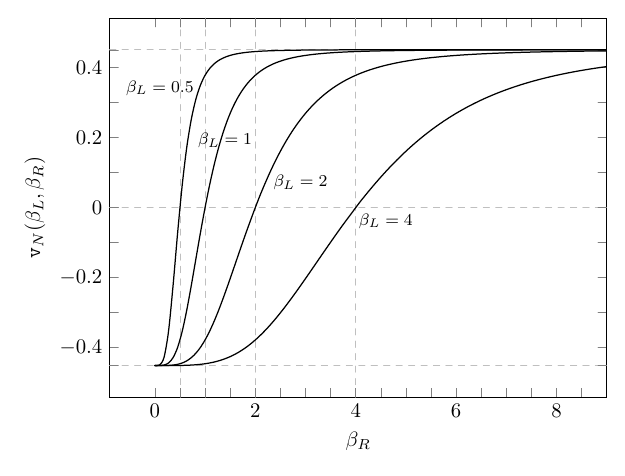}
		\caption{$\beta_R \mapsto \vel_N (\beta_L , \beta_R)$}
		\label{fig:sub1-ness-comoving}
	\end{subfigure}%
	\begin{subfigure}{.5\textwidth}
		\centering
		\includegraphics[width=\textwidth]{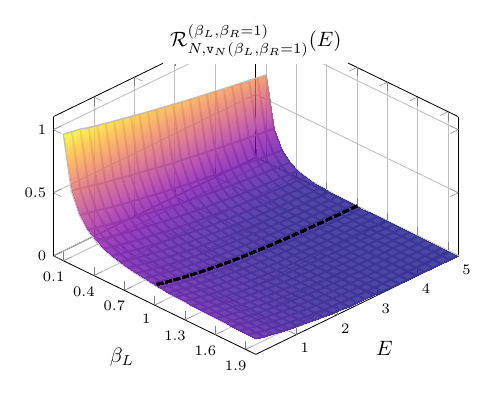}
		\caption{$(\beta_L , E)\mapsto\Rcal_{N,\vel_N (\beta_L , \beta_R = 1)}^{(\beta_L , \beta_R = 1)} (E)$ ($\beta_L > 0.1$)}
		\label{fig:sub2-ness-comoving}
	\end{subfigure}
	\caption{\textbf{(a)} Plot of the velocity $\vel_N (\beta_L , \beta_R)$ of the comoving frame as a function of $\beta_R$ for different fixed values of $\beta_L$. The velocity $\vel_N (\beta_L , \beta_R)$ is negative when $\beta_R < \beta_L$, and positive when $\beta_R > \beta_L$. The modulus of the velocity is bounded by $\frac{4-\sqrt{7}}{3} \approx 0.4514$. The velocity function $\vel_N$ extends to $(\beta_L , \beta_R)$ with $\beta_R = \beta_L$, where it vanishes. The functions $\beta_L \mapsto \vel_N (\beta_L , \beta_R)$ for fixed $\beta_R$ are obtained via reflection along the horizontal axis. \quad \textbf{(b)} Surface plot of $(\beta_L , E)\mapsto\Rcal_{N,\vel}^{(\beta_L , \beta_R)} (E)$ for $\vel=\vel_N (\beta_L , \beta_R)$, $\beta_R = 1$, and $\beta_L > 0.1$, $E>0$. The dashed surface curve at $\beta_L = 1$ (in which case the detector interacts with a $(\beta=1)$-KMS state and $\vel_N \to 0$) represents the Planckian $\Rcal_{0}^{(\beta=1)}$.}
	\label{fig:rate-ness-comoving}
\end{figure}

This is in contrast to the NESS $\sigma_N$ of the massless scalar field on two-dimensional Minkowski spacetime, for which $\sigma_N (T^{00}) = \sigma_N (T^{11}) = \frac{\pi}{12} (\beta_L^{-2} + \beta_R^{-2})$ and $\sigma_N (T^{01}) = \sigma_N (T^{10}) = \frac{\pi}{12} (\beta_L^{-2} - \beta_R^{-2})$ (see \cite{DLSB}). 
The velocity of the boost that diagonalizes the expected stress-energy tensor in that case is given by 
\begin{gather*}
	\vel_1 := \kappa_1 - \sgn(\beta_R - \beta_L) \sqrt{\kappa_1^2 - 1} = \frac{\beta_R - \beta_L}{\beta_R + \beta_L} \, , \quad \kappa_1 := \frac{\sigma_N (T^{00})}{\sigma_N (T^{01})} = \frac{\beta_R^2 + \beta_L^2}{\beta_R^2 - \beta_L^2} \, ,
\end{gather*}
which is precisely the boost velocity for the frame in which $\sigma_N$ is a KMS state of temperature $(\beta_L \beta_R)^{-1/2}$ \cite{DLSB} (see also \cite{BDLS2015}). 
The modulus of $\vel_1$ approaches the speed of light in the ground state limit of one of the heat baths.\medskip

As noted at the beginning of this section, we propose that the significance of the comoving frame for the study of detectors coupled to NESS arises from the paradigm that the instantaneous rest frame of a system is the preferred reference frame for the formulation of its thermal properties. 
For the NESS with initial heat bath KMS parameters $\beta_L , \beta_R$ the comoving frame with velocity $\vel_N (\beta_L , \beta_R)$ is the inertial rest frame of the heat flow of the NESS. 
In the case \textsf{(NESS-$\vel_N$)} in which the detector moves along the $x^1$-axis with velocity $\vel_N (\beta_L , \beta_R)$ relative to the inertial frame $I$, the transition rate is given by $\Rcal_{N,\vel_N (\beta_L , \beta_R)}^{(\beta_L , \beta_R)}$ (Eq.\ \eqref{eq:transition-rate-ness-vel}), which is now a function of $\beta_L , \beta_R , E$. 
The transition rate $\Rcal_{N,\vel_N (\beta_L , \beta_R)}^{(\beta_L , \beta_R)}$ is kept invariant when the KMS parameters $\beta_L , \beta_R$ are swapped, thanks to the property $\vel_N (\beta_L , \beta_R)=-\vel_N (\beta_R , \beta_L)$ (cf.\ the discussion below Proposition \ref{prop:ness-v}). 
The heat flow is absent and the relative velocity of the comoving frame vanishes when $\beta_L = \beta_R =: \beta$ (see Figure \ref{fig:sub1-ness-comoving}), so the transition rate $\Rcal_{N,\vel_N (\beta_L , \beta_R)}^{(\beta_L , \beta_R)}$ of the detector reduces to $\Rcal_{0}^{(\beta)}$, the Planckian transition rate of the detector at rest relative to the KMS state (see Eqs.\ \eqref{eq:transition-rate-ness-vel} \& \eqref{eq:transition-rate-kms-rest}). 
In a sense, choosing the relative velocity $\vel$ of the detector to be the temperature-dependent velocity $\vel_N$ provides an interpolation between the transition rates of the cases \textsf{(NESS-$\vel$)} and \textsf{(KMS-$0$)} as the parameters $\beta_L,\beta_R$ are varied. 

Figure \ref{fig:sub2-ness-comoving} shows a plot of $(\beta_L , E) \mapsto \Rcal_{N,\vel_N (\beta_L , 1)}^{(\beta_L , 1)} (E)$ for $E>0$. 
Although the detector is at rest relative to the heat flow (the ``medium'' of the NESS), the response is non-thermal (unless $\beta_L = \beta_R$), thus showing a difference to the case of a detector that couples to the rest frame of an inertial heat bath. 
This is in line with the anisotropy of the expected stress-energy tensor in the comoving frame (Proposition \ref{prop:comoving}): 
While the transformation to the comoving frame compensates the impact of the heat flow in longitudinal direction (so the expected stress-energy tensor is diagonalized), transversal modes still transmit (Doppler shifted) energy and momentum to the detector under the isotropic monopole coupling. 
Further studies of the response with other detector couplings could reveal more details regarding this behavior.

\section{Conclusions and outlook}
\label{sec:conclusions-outlook}

In this work we studied the transition rate of a two-level Unruh-DeWitt detector that is coupled via a point-like monopole interaction to the NESS of a free massless scalar field on four-dimensional Minkowski spacetime. 
The NESS, which has been constructed in \cite{DLSB,Hack-Verch}, arises from bringing two semi-infinite heat baths, described by inertial KMS states at inverse temperatures $\beta_L , \beta_R >0$, into contact in a neighborhood of the spatial surface with $x^1 = 0$. 
The state is invariant under time and space translations and features a steady flow of energy between the heat baths. 
The right- and left-moving modes along the $x^1$-axis are in thermal equilibrium at inverse temperatures $\beta_L$ and $\beta_R$, respectively. 

Accordingly, if the detector is at rest relative to the defining inertial (laboratory) frame of the NESS (Section \ref{sec:case-NESS-stat}), the transition rate equals the arithmetic mean of the Planckian rates associated to the two heat baths (Proposition \ref{prop:ness-stat}). 
The NESS looks to the detector like a mixture of KMS states and thus the response corresponds to an interaction with a state that is passive \cite{Pusz-Woronowicz,Fewster-Verch2003} (see also \cite[Sec.\ 4.8.2]{Fewster-Verch-AQFT2015}) but not in thermal equilibrium at some uniquely assigned temperature. 
For non-zero constant velocity along the $x^1$-axis (Section \ref{sec:case-NESS-vel}) the form of the detector's transition rate (Proposition \ref{prop:ness-v}) indicates the impact of purely kinematical effects (Doppler effect). 
This is structurally similar to the response obtained from the motion relative to a thermal bath \cite{Costa-Matsas-background1995,Costa-Matsas1995,Landsberg-Matsas1996}, reviewed in Section \ref{sec:case-KMS}. 
As a special case we considered the limit $\beta_L \to \infty$ that brings the left reservoir into the vacuum state, and discussed the dependence of the corresponding transition rate on the relative velocity of the detector. 
In particular, a small velocity expansion allowed to highlight differences to the case of a detector moving in a heat bath \cite{Costa2004} (Corollary \ref{cor:ness-quotient}). 
In Section \ref{sec:detailed-balance} we collected some properties of the corresponding effective temperatures obtained from the detailed balance condition (Proposition \ref{prop:eff-temp}). 
These results may represent a starting point for further research of NESS being probed by Unruh-DeWitt detectors, as there are several topics we have not pursued here (e.g., numerical analysis of transition rates and effective temperatures; ramifications with regards to the comoving inertial frame that moves with the heat flow of the NESS; influence of other detector couplings, spacetime dimensions, and non-zero mass parameter of the field). \medskip

The main takeaway from our results is that a monopole detector coupled to a NESS is, in general, not sensitive to dynamical features. 
A proper ``NESS detector'' should be able to pick up the current of the heat flow and convert its energy-momentum into internal energy, which then would be visible in the transition rate already for a detector at rest. 
In case such a detector moves along with the heat flow by being at rest relative to the comoving frame (as introduced in Section \ref{sec:case-NESS-comoving}), one would expect that the response corresponds to a passive state, since in this frame the current is compensated and thus one cannot extract energy from it via cyclic processes. 
(To try the analogy employed in \cite{Buchholz-Solveen}, the detector is stationary in the rest frame of the spinning ``pinwheel'', characterized by zero flow velocity from the detector's perspective.) 
All of this certainly depends on the choice of detector model and detector-field interaction. 
Couplings of interest could involve higher moment operators; see, e.g., \cite{Hinton1983,Moustos2018,Papadatos-Anastopoulos2020,Takagi1986,Marzlin-Audretsch1998} (and references therein) for some works that consider dipole moments in the interaction with a derivative of a scalar field or in the context of model atoms coupled to an electromagnetic field. 
More specifically, a dynamical coupling to the heat flow might be achieved by means of a dipole (or multipole) moment that couples to tensorial, current-related quantities of the NESS, like the stress-energy tensor (see \cite{Padmanabhan-Singh1987} for such a model). 
This could facilitate the study of (an)isotropy of the heat flow of the NESS from the perspective of the detector.\medskip

Finally, our discussion reflects the well-known subtlety of the concept of temperature in non-equilibrium and relativistic systems \cite{CasasVazquez-Jou2003,Neugebauer}. 
In certain instances it might be sensible to assign an effective temperature to a relativistic NESS. 
For example, one can introduce a Lorentz-invariant ``proper effective temperature'' based on gauge-gravity duality \cite{Hoshino-Nakamura2020} (see also references therein). 
However, the non-Planckian transition rates in Sections \ref{sec:case-NESS-stat} \& \ref{sec:case-NESS-vel} and the corresponding (energy-dependent) detailed balance temperatures obtained in Section \ref{sec:detailed-balance} reveal that generally such a quantity will not be a temperature by way of Unruh-DeWitt detectors, i.e.\ it is not a temperature from a microscopic (Boltzmann-Gibbs or KMS equilibrium) viewpoint (cf.\ \cite{Haag1996}). 
As far as the response of the detector is concerned, the NESS does not have a uniquely determined temperature, similar to heat baths from the perspective of moving inertial reference frames \cite{Costa-Matsas1995,Landsberg-Matsas1996,Landsberg-Matsas2004}. 
The mentioned limitation of the detector setup to measure dynamical properties of the NESS motivates a deeper investigation of these topics.


\bigskip
\bigskip

\mysepline

\paragraph{Acknowledgments}

The authors thank Marco Merkli for discussion and pointing out several references regarding NESS, and two anonymous referees for helpful comments. 
A.G.P.\ is indebted to the IMPRS and MPI for Mathematics in the Sciences, Leipzig, where early parts of this work have been conceived, for support and hospitality. 
R.V.\ gratefully acknowledges funding by the CY Initiative of Excellence (grant ``Investissements d'Avenir'' ANR-16-IDEX-0008) during his stay at the CY Advanced Studies.

\bigskip
\bigskip

\mysepline


\appendix

\section{Transition rate of detector moving through heat bath}
\label{appendix:cm}

In this appendix we prove Eq.\ \eqref{eq:transition-rate-kms-moving}, the transition rate of an Unruh-DeWitt monopole detector moving with constant non-zero velocity while coupled to a free massless scalar field on four-dimensional Minkowski spacetime at positive temperature. 
The result has been stated before by Costa \& Matsas \cite{Costa-Matsas-background1995,Costa-Matsas1995}, and a proof by contour integration has been outlined in \cite[Appendix C.2]{Hodgkinson-Louko-Ottewill2014}. 
(For a detector coupled to the time derivative of a scalar field on three-dimensional Minkowski spacetime a corresponding result can be found in \cite{Bunney-Louko2023}.) 
Serving as reference in the main text, we present a derivation starting from the two-point function $W_\beta$ of the field's $\beta$-KMS state $\omega_\beta$ ($\beta>0$) given by Eq.\ \eqref{eq:thermal-two-point}. \medskip 

In accordance with the Hadamard property, $\widetilde{W}_\beta := W_\beta - W_{\mathrm{vac}}$ (for the vacuum two-point function $W_{\mathrm{vac}}$, see Eq.\ \eqref{eq:vac-two-point}) is a smooth function on $\M\times\M$. 
It can be written in the form (see \cite{Weldon2000} and \cite[Sec.\ F.3.5]{Frolov2011})
\begin{align}
	\label{eq:thermal-two-point2}
	\widetilde{W}_\beta(x,y) & = \frac{1}{4\pi^2 \left((x^0 - y^0)^2 - \lVert\xvec-\yvec\rVert^2 \right)} + \\ & + \frac{1}{8\pi\beta\lVert\xvec-\yvec\rVert} \left[ \coth\left( \frac{\pi}{\beta} \left(x^0 - y^0 + \lVert\xvec-\yvec\rVert \right) \right) - \coth\left( \frac{\pi}{\beta} \left(x^0 - y^0 - \lVert\xvec-\yvec\rVert\right) \right) \right] \, . \nonumber
\end{align}
For the sake of self-containment, let us derive this expression from the distributional integral representations of $W_{\mathrm{vac}}$ and $W_\beta$. 
As these two-point functions only depend on $x-y$ we can set $y=0$ without loss of generality. 
Using spherical coordinates,
\begin{align*}
	\widetilde{W}_\beta (x,0) &= \frac{1}{(2\pi)^3} \int\limits_{\R^3} \frac{1}{2\lVert\pvec\rVert}  \frac{\e^{i\lVert\pvec\rVert x^0} \e^{-i\pvec\cdot\xvec} + \e^{-i\lVert\pvec\rVert x^0} \e^{i\pvec\cdot\xvec}}{\e^{\beta\lVert\pvec\rVert} - 1} \, \Diff3\pvec = \\ &= \frac{1}{(2\pi)^2} \int_0^\infty \int_0^\pi  \frac{r\cos\left(r x^0 - r\lVert\xvec\rVert \cos(\theta)\right)}{\e^{\beta r} - 1} \, \sin(\theta) \diff\theta \, \diff r \, .
\end{align*}
For $\xvec=\vec{0}$ we have $\widetilde{W}_\beta ((x^0 , \vec{0}),0) = \frac{1}{4\pi^2 (x^0)^2} - \frac{1}{4\beta^2} \frac{1}{\sinh^2 (\frac{\pi}{\beta} x^0 )}$ by \cite[Sec.\ 1.4, (8)]{Erdelyi}, which leads to $\widetilde{W}_\beta (0,0)=\frac{1}{12\beta^2} = \widetilde{W}_\beta (x,x)$ (for all $x\in\M$) in the limit $x^0 \to 0$ (see also \cite[Sec.\ 6.3, (7)]{Erdelyi}). 
For $\xvec\neq\vec{0}$, 
\begin{align*}
	\widetilde{W}_\beta (x,0) &= \frac{1}{4\pi^2 \lVert\xvec\rVert} \sum\limits_{\varsigma=\pm 1} \varsigma \int_0^\infty \frac{\sin\left(r\left(x^0 + \varsigma \lVert\xvec\rVert\right)\right)}{\e^{\beta r} - 1} \, \diff r = \\ &= \frac{1}{4\pi^2 \lVert\xvec\rVert} \sum\limits_{\varsigma=\pm 1} \varsigma \left( -\frac{1}{2 \left(x^0 + \varsigma\lVert\xvec\rVert\right)} + \frac{\pi}{2\beta} \coth\left(\frac{\pi}{\beta} \left(x^0+\varsigma\lVert\xvec\rVert\right) \right) \right) \, ,
\end{align*}
which extends smoothly to all $x\in\M$ and results in Eq.\ \eqref{eq:thermal-two-point2} after reinstating $y$ by $\widetilde{W}_\beta (x,y)=\widetilde{W}_\beta (x-y,0)$. 
Here, we used the Fourier sine transform \cite[Sec.\ 2.4, (11)]{Erdelyi}
\begin{gather}
	\int_0^\infty \frac{\sin(qr)}{\e^{\beta r} - 1} \, \diff r = -\frac{1}{2q} + \frac{\pi}{2\beta} \coth\left(\frac{\pi q}{\beta} \right) \, , \quad \beta , q > 0 \, ,
	\label{eq:sine-transf}
\end{gather}
which also applies to $q<0$ by symmetry and extends smoothly to $q=0$.\medskip 

Now let $\widetilde{w}_{\beta,\vel} (s):=\widetilde{W}_\beta (\xsf_{\vel} (s),0) \equiv (W_\beta - W_{\mathrm{vac}})(\xsf_{\vel} (s),0)$ for the inertial detector worldline $\tau\mapsto \xsf_{\vel} (\tau)=\gam (\tau , \vel \tau ,0,0)$ with $0<|\vel|<1$. 
(Any other inertial worldline, with constant velocity $\vel$ in a different spatial direction, leads to the same transition rate.) 
The transition rate (Eq.\ \eqref{eq:transition-rate}) in question is 
\begin{gather}
	\label{eq:transition-rate-kms-moving2}
	\Rcal_{\vel}^{(\beta)} (E) = \int\limits_{\R} \e^{-iEs} \widetilde{w}_{\beta,\vel} (s) \diff s - \frac{E}{2\pi} \Theta(-E)
\end{gather}
for $E\in\R\setminus\{0\}$, where we split off the vacuum contribution that is given by Eq.\ \eqref{eq:transition-rate-vac}. 
Using Eq.\ \eqref{eq:thermal-two-point2} one finds that
\begin{align}
	\widetilde{w}_{\beta,\vel} (s) &= \frac{1}{4\pi^2 s^2} + \frac{1}{8\pi\beta\gam|\vel s|} \left( \coth\left( \frac{\pi \gam}{\beta} (s+|\vel s|) \right) - \coth\left( \frac{\pi \gam}{\beta} (s-|\vel s|) \right) \right) = \nonumber \\ &= \frac{1}{4\pi\beta\gam\vel} \left[ \frac{\beta\gam\vel}{\pi s^2} + \frac{1}{2s} \coth\left( \frac{\pi s}{\beta \dop_-} \right) - \frac{1}{2s} \coth\left( \frac{\pi s}{\beta \dop_+} \right) \right] = \nonumber \\ &= \frac{1}{4\pi\beta\gam\vel} \left( f_{\beta \dop_+} (s) - f_{\beta \dop_-} (s) \right) \, , \label{eq:two-point-reg}
\end{align}
where $\dop_\pm = \gam(1\pm\vel)$ are the Doppler factors, and
\begin{gather}
	\label{eq:fc}
	f_c (s) := \frac{c}{2\pi s^2} - \frac{1}{2s} \coth\left( \frac{\pi s}{c} \right) = -\frac{c}{\pi} \sum\limits_{n=1}^\infty \frac{1}{n^2 c^2 + s^2} \, , \quad c>0 \, .
\end{gather}
Note that the Laurent series expansion of $s\mapsto \frac{1}{2s} \coth(\frac{\pi s}{c})$ in a punctured neighborhood of $s=0$ has the singular part $\frac{c}{2\pi s^2}$, and $\lim_{s\to 0} f_c (s)=-\frac{\pi}{6c}$. 
For every $c>0$ the function $f_c$ is defined as the unique extension of the given expression to a bounded, even, smooth function on $\R$ as signified by the stated series expansion in Eq.\ \eqref{eq:fc} (see \cite[4.36.3]{NIST:DLMF}). 
The Fourier cosine transform $\int_0^\infty \cos(sx) \ln(1-\e^{-cx}) \diff x = \pi f_c (s)$ (for $c,s>0$) from \cite[Sec.\ 1.5, (14)]{Erdelyi} can be inverted to $\int_0^\infty \cos(qs) f_c (s) \diff s = \frac{1}{2} \ln(1-\e^{-cq})$ for $c,q>0$, and thus
\begin{gather}
	\label{eq:fc-fourier}
	\int\limits_{\R} \e^{-iEs} f_c (s) \diff s = 2\int_0^\infty \cos(|E|s) f_c (s) \diff s = \ln(1-\e^{-c|E|}) \, .
\end{gather}
Combining Eqs.\ \eqref{eq:two-point-reg} \& \eqref{eq:fc-fourier} we obtain, for $E\in\R\setminus\{0\}$, 
\begin{gather*}
	\int\limits_{\R} \e^{-iEs} \widetilde{w}_{\beta,\vel} (s) \diff s = \frac{1}{4\pi\beta\gam\vel} \ln\left(\frac{1-\e^{-\beta \dop_+ |E|}}{1-\e^{-\beta \dop_- |E|}} \right) \, .
\end{gather*}
For $E<0$ this gives 
\begin{gather*}
	\int\limits_{\R} \e^{-iEs} \widetilde{w}_{\beta,\vel} (s) \diff s = \frac{1}{4\pi\beta\gam\vel} \ln\left(\frac{1-\e^{\beta \dop_+ E}}{1-\e^{\beta \dop_- E}} \right) = \frac{(\dop_+ - \dop_-)E}{4\pi\gam\vel} + \frac{1}{4\pi\beta\gam\vel} \ln\left(\frac{1-\e^{-\beta \dop_+ E}}{1-\e^{-\beta \dop_- E}} \right) \, ,
\end{gather*}
of which the first term equals $\frac{E}{2\pi}$ since $\dop_+ - \dop_- = 2\gam\vel$. 
Eq.\ \eqref{eq:transition-rate-kms-moving2} thus results in Eq.\ \eqref{eq:transition-rate-kms-moving}.

\section[Alternative derivation for NESS-v]{Alternative derivation for \textsf{(NESS-$\vel$)}}
\label{appendix:NESS-v}

Complementary to the proof of Proposition \ref{prop:ness-v} we provide an alternative derivation for $\widetilde{w}_{N,\vel}$, the smooth part of the NESS two-point function pulled back along the inertial worldline $\xsf_{\vel}$ for $\vel\in(-1,1)\setminus\{0\}$ (Eq.\ \eqref{eq:inertial-worldline}), by presenting explicit expressions for the integrals obtained from the two terms in Eq.\ \eqref{eq:ness-two-point-vel}. 
Using spherical coordinates we see that 
\begin{align*}
	& \int\limits_{\R^3} \frac{1}{2\lVert\pvec\rVert} \frac{\e^{ip_1 \gam \vel s} \e^{i\lVert\pvec\rVert \gam s}}{\e^{\upbeta(-p_1)\lVert\pvec\rVert} - 1} \, \Diff3\pvec = \\ & = \pi \int_0^\infty \int_0^{\pi/2} r \, \left(\frac{\e^{i\gam\vel sr \cos(\theta)} \e^{i\gam sr}}{\e^{\beta_R r}-1} + \frac{\e^{-i\gam\vel sr \cos(\theta)} \e^{i\gam sr}}{\e^{\beta_L r}-1} \right) \, \sin(\theta) \diff\theta \diff r  = \\ & = \frac{-i\pi}{\gam\vel s} \int_0^\infty \frac{\e^{i\gam sr} (\e^{i\gam\vel sr} - 1)}{\e^{\beta_R r}-1} \, \diff r + \frac{i\pi}{\gam\vel s} \int_0^\infty \frac{\e^{i\gam sr} (\e^{-i\gam\vel sr} - 1)}{\e^{\beta_L r}-1} \, \diff r = \\ & = \frac{-i\pi}{\gam\vel\beta_R s} \left[ \uppsi\left(1-\frac{i\gam s}{\beta_R}\right) - \uppsi\left(1-\frac{i\dop_+ s}{\beta_R}\right) \right] + \frac{i\pi}{\gam\vel\beta_L s} \left[ \uppsi\left(1-\frac{i\gam s}{\beta_L}\right) - \uppsi\left(1-\frac{i\dop_- s}{\beta_L}\right) \right]
\end{align*}
for all $s\in\R$ (with extension to $s=0$ implied), where $\uppsi$ is the digamma function (logarithmic derivative of the gamma function) \cite[5.2.2]{NIST:DLMF}. 
Here we used $\int_0^\infty \frac{\e^{iax} - \e^{ibx}}{\e^x - 1} \diff x = \uppsi(1-ib)-\uppsi(1-ia)$ for $a,b\in\R$, which follows from the integral representation $\uppsi(z)=\int_0^\infty \left( \frac{\e^{-x}}{x} - \frac{\e^{-(z-1)x}}{\e^x - 1} \right) \diff x$ for $z\in\C$, $\Re(z)>0$ \cite[5.9.12]{NIST:DLMF}. 
Similarly, 
\begin{align*}
	& \int\limits_{\R^3} \frac{1}{2\lVert\pvec\rVert} \frac{\e^{ip_1 \gam \vel s} \e^{-i\lVert\pvec\rVert \gam s}}{\e^{\upbeta(p_1)\lVert\pvec\rVert} - 1} \, \Diff3\pvec = \\ & = \frac{-i\pi}{\gam\vel\beta_L s} \left[ \uppsi\left(1+\frac{i\gam s}{\beta_L}\right) - \uppsi\left(1+\frac{i\dop_- s}{\beta_L}\right) \right] + \frac{i\pi}{\gam\vel\beta_R s} \left[ \uppsi\left(1+\frac{i\gam s}{\beta_R}\right) - \uppsi\left(1+\frac{i\dop_+ s}{\beta_R}\right) \right] \, .
\end{align*}
Applying these results to Eq.\ \eqref{eq:ness-two-point-vel} and using $\uppsi(1+ia)-\uppsi(1-ia) = i\pi\coth(\pi a)-\frac{i}{a}$ for $a\in\R\setminus\{0\}$ \cite[5.5.2 \& 5.5.4]{NIST:DLMF}, one obtains $\widetilde{w}_{N,\vel} (s)$ for $s\in\R\setminus\{0\}$ as in Eq.\ \eqref{eq:ness-two-point-vel-2}, which extends smoothly to zero with $\widetilde{w}_{N,\vel} (0) = \frac{1}{24} (\beta_L^{-2} + \beta_R^{-2})$.

\section{Comments on effective temperatures}
\label{appendix:eff-temp}

Consider an inertial observer moving with constant velocity $\vel\in(-1,1)$ relative to a heat bath of temperature $\beta^{-1} > 0$. 
The anisotropic radiation energy spectrum observed in some fixed direction forming an angle $\theta$ to the axis of motion has been found, e.g., in \cite{Peebles-Wilkinson,Bracewell-Conklin,Henry-CMB1968} in the context of the cosmic microwave background radiation as observed on Earth, and is given by $f(E,T') := E/(2\pi(\e^{E/T'(\beta,\vel,\theta)} -1))$ with the \emph{directional temperature} (see also \cite{Landsberg-Matsas1996,Landsberg-Matsas2004}) 
\begin{gather*}
	T'(\beta,\vel,\theta)=\frac{\sqrt{1-\vel^2}}{\beta(1-|\vel|\cos(\theta))} \, .
\end{gather*}
Along the axis of motion ($\theta=0$) this equals the Doppler shifted temperature $T'(\beta,\vel,0)=\beta^{-1} \sqrt{(1+|\vel|)/(1-|\vel|)}$. 
A critical assessment and derivation of $T'$ can be found in \cite{Nakamura2009}. 
Averaging $f(E,T')$ over solid angles results in the distribution given in Eq.\ \eqref{eq:transition-rate-kms-moving} that appears as the transition rate of an Unruh-DeWitt detector carried by the moving observer and coupled to the heat bath of a massless scalar field \cite{Costa-Matsas1995,Landsberg-Matsas1996}. 
The solid angle average of $T'$ is (see \cite[Eq.\ (8)]{Landsberg-Matsas1996})
\begin{gather*}
	\frac{1}{4\pi} \int_0^{2\pi} \int_0^\pi T'(\beta,\vel,\theta) \sin(\theta) \diff\theta \diff\varphi = \frac{\sqrt{1-\vel^2}}{2\beta} \int_0^\pi \frac{\sin(\theta)}{1-|\vel|\cos(\theta)} \diff\theta = \frac{\sqrt{1-\vel^2}}{2\beta\vel} \ln\left(\frac{1+\vel}{1-\vel}\right) \, ,
\end{gather*}
which equals the lowest order of the detailed balance effective temperature in the case \ref{item:KMS} (see Eq.\ \eqref{eq:Tvbeta}).

\section{The NESS stress-energy tensor}
\label{appendix:set-ness}

We show that the expectation value of the stress-energy tensor $T_{\mu\nu}$ (Eq.\ \eqref{eq:set}) in the NESS $\omega_N$ of the free massless scalar field on four-dimensional Minkowski spacetime (Definition \ref{def:ness}) has block diagonal form and calculate its components. 
The results have been presented before in \cite{DLSB} (for $T^{\mu\nu}$) for free scalar fields on $(1+d)$-dimensional Minkowski spacetime ($d\geq 1$) and with arbitrary field mass. 
For completeness, and to connect with the definitions used in our work, we rederive these results using the two-point function of the NESS based on \cite{Hack-Verch}.

\begin{lemma}[NESS expected stress-energy tensor]
	\label{lem:set-ness-components}
	The expectation value of the stress-energy tensor $T_{\mu\nu}$ of the massless scalar field on four-dimensional Minkowski spacetime $\M$ in the NESS $\omega_N$ (with $\beta_L , \beta_R > 0$) has the matrix form
	\begin{gather*}
		(\Tcal_{\mu\nu}) := (\omega_N (T_{\mu\nu} (x))) = \begin{pmatrix} \widetilde{\Tcal} & \nullmatrix_2 \\ \nullmatrix_2 & \diag (\Tcal_{\perp},\Tcal_{\perp}) \end{pmatrix}
	\end{gather*}
	for every $x\in\M$, where $\nullmatrix_2$ is the $2\times 2$ zero matrix, $\widetilde{\Tcal}$ is a symmetric, non-zero $2\times 2$ matrix, and $\Tcal_{\perp}:=\Tcal_{22} = \Tcal_{33}$. 
	The non-trivial components are given by
	\begin{gather*}
		\Tcal_{00} = \frac{\pi^2}{60} \left( \frac{1}{\beta_L^4} + \frac{1}{\beta_R^4} \right) \, , \quad \Tcal_{01} = \frac{\pi^2}{120} \left( \frac{1}{\beta_R^4} - \frac{1}{\beta_L^4} \right) \, , \quad \Tcal_{ii} = \frac{1}{3} \Tcal_{00} \quad (i\in\{1,2,3\}) \, .
	\end{gather*}
\end{lemma}

\begin{proof}
	Let $p_0 := \lVert\pvec\rVert \equiv (\sum_{i=1}^3 p_i^2)^{1/2}$ and $\upbeta(p_1) := \Theta(p_1) \beta_L + \Theta(-p_1)\beta_R$. 
	From the definition of the (renormalized) stress-energy tensor (Eq.\ \eqref{eq:set}) it follows that
	\begin{gather*}
		\Tcal_{\mu\nu} (x) := \omega_N (T_{\mu\nu} (x)) = \left[ \left( \partial_\mu \otimes \partial_\nu - \frac{1}{2} \eta_{\mu\nu} \partial_\rho \otimes \partial^\rho \right) \widetilde{W}_N \right] (x,x)
	\end{gather*}
	for the smooth function $\widetilde{W}_N = W_N - W_{\mathrm{vac}}$ given by Eq.\ \eqref{eq:ness-two-point-reg}, where Einstein's summation convention is implied, the first (second) tensor factor applies to the first (second) argument variable $x$ ($y$), and eventually the coincidence limit $y\to x$ is taken. 
	Since $\partial_0 \otimes \partial_0 \, \e^{\pm i\lVert\pvec\rVert(x^0 - y^0)} \restr_{x=y} \, = \lVert\pvec\rVert^2$ and $\partial_i \otimes \partial_j \, \e^{i\pvec\cdot(\xvec-\yvec)} \restr_{x=y} \, = p_i p_j$ (for $i,j\in\{1,2,3\}$) we get
	\begin{align*}
		\left( \partial_\mu \otimes \partial_\mu \widetilde{W}_N  \right)(x,x) &= \frac{1}{(2\pi)^3} \int\limits_{\R^3} \frac{1}{2\lVert\pvec\rVert} p_\mu^2 \left( \frac{1}{\e^{\upbeta(-p_1)\lVert\pvec\rVert} - 1} + \frac{1}{\e^{\upbeta(p_1)\lVert\pvec\rVert} - 1} \right) \, \Diff3\pvec = \\ &= \frac{1}{(2\pi)^3} \int\limits_{\R^3} \frac{1}{2\lVert\pvec\rVert} \frac{2p_\mu^2}{\e^{\upbeta(p_1)\lVert\pvec\rVert} - 1} \, \Diff3\pvec
	\end{align*}
	for $\mu\in\{0,1,2,3\}$. 
	For the calculations we use $\int_0^\infty \frac{r^3}{\e^{\beta r} -1} \diff r = \frac{\pi^4}{15\beta^4}$ \cite[Sec.\ 6.3, (7)]{Erdelyi}. 
	Since the expressions are independent of $x\in\M$ we simply write $\Tcal_{\mu\nu}$ from now on. 
	It follows that
	\begin{align*}
		\Tcal_{00} &= \frac{1}{2} \left[ \left( \partial_0 \otimes \partial_0 + \sum\limits_{i=1}^3 \partial_i \otimes \partial_i \right) \widetilde{W}_N \right] (x,x) = \\ &= \frac{1}{2} \frac{1}{(2\pi)^3} \int\limits_{\R^3} \frac{1}{2\lVert\pvec\rVert} \frac{2p_0^2 + 2\sum_{i} p_i^2}{\e^{\upbeta(p_1)\lVert\pvec\rVert} - 1} \, \Diff3\pvec = \frac{1}{(2\pi)^3} \int\limits_{\R^3} \frac{\lVert\pvec\rVert}{\e^{\upbeta(p_1)\lVert\pvec\rVert} - 1} \Diff3\pvec = \\ &= \frac{1}{(2\pi)^3} 2\pi \int_0^\infty \frac{r}{\e^{\beta_L r} - 1} \, r^2 \diff r + \frac{1}{(2\pi)^3} 2\pi \int_0^\infty \frac{r}{\e^{\beta_R r} - 1} \, r^2 \diff r = \frac{\pi^2}{60} \left( \frac{1}{\beta_L^4} + \frac{1}{\beta_R^4} \right) \, ,
	\end{align*}
	where the integral over $p_1$ is split according to the definition of $\upbeta(p_1)$, and the resulting integrals are calculated in spherical coordinates on $(0,\infty) \times \R^2$. 
	Likewise, for $i\in\{1,2,3\}$,
	\begin{align*}
		\Tcal_{ii} &= \frac{1}{2} \left[ \left( \partial_i \otimes \partial_i + \partial_0 \otimes \partial_0 - \sum\limits_{j\neq i} \partial_j \otimes \partial_j \right) \widetilde{W}_N \right] (x,x) = \\ &= \frac{1}{2} \frac{1}{(2\pi)^3} \int\limits_{\R^3} \frac{1}{2\lVert\pvec\rVert} \frac{2p_i^2 + 2p_0^2 - 2\sum_{j\neq i} p_j^2}{\e^{\upbeta(p_1)\lVert\pvec\rVert} - 1} \, \Diff3\pvec = \frac{1}{(2\pi)^3} \int\limits_{\R^3} \frac{1}{2\lVert\pvec\rVert} \frac{2p_i^2}{\e^{\upbeta(p_1)\lVert\pvec\rVert} - 1} \, \Diff3\pvec = \\ &= \frac{1}{(2\pi)^3} \frac{2\pi}{3} \int_0^\infty r^3 \left(\frac{1}{\e^{\beta_L r} - 1} + \frac{1}{\e^{\beta_R r} - 1} \right) \, \diff r = \frac{\pi^2}{180} \left( \frac{1}{\beta_L^4} + \frac{1}{\beta_R^4} \right) \equiv \frac{1}{3} \Tcal_{00} \, .
	\end{align*}
	
	Since $\partial_0 \otimes \partial_i \, \e^{i\pvec\cdot(\xvec-\yvec)} \e^{\pm i\lVert\pvec\rVert(x^0 - y^0)} \restr_{x=y} \, = \pm \lVert\pvec\rVert p_i \equiv \pm p_0 p_i$ (for $i\in\{1,2,3\}$),
	\begin{align*}
		\Tcal_{i0} &= \Tcal_{0i} = \left( \partial_0 \otimes \partial_i \widetilde{W}_N  \right)(x,x) = \\ &= \frac{1}{(2\pi)^3} \int\limits_{\R^3} \frac{1}{2\lVert\pvec\rVert} p_0 p_i \left( \frac{1}{\e^{\upbeta(-p_1)\lVert\pvec\rVert} - 1} - \frac{1}{\e^{\upbeta(p_1)\lVert\pvec\rVert} - 1} \right) \, \Diff3\pvec \, ,
	\end{align*}
	which vanishes identically except for
	\begin{align*}
		\Tcal_{01} &= \frac{1}{(2\pi)^3} \int\limits_{\{p_1 > 0\}\times\R^2} p_1 \left(\frac{1}{\e^{\beta_R \lVert\pvec\rVert} - 1} - \frac{1}{\e^{\beta_L \lVert\pvec\rVert} - 1} \right) \Diff3\pvec = \\ &= \frac{1}{(2\pi)^3} \pi \int_0^\infty r^3 \left(\frac{1}{\e^{\beta_R r} - 1} - \frac{1}{\e^{\beta_L r} - 1} \right) \, \diff r = \frac{\pi^2}{120} \left( \frac{1}{\beta_R^4} - \frac{1}{\beta_L^4} \right) \, .
	\end{align*}
	
	Finally, for $i\in\{2,3\}$ one observes that 
	\begin{align*}
		\Tcal_{i1} &= \Tcal_{1i} = \left( \partial_1 \otimes \partial_i \widetilde{W}_N  \right)(x,x) = \\ &= \frac{1}{(2\pi)^3} \int\limits_{\R^3} \frac{1}{2\lVert\pvec\rVert} p_1 p_i \left( \frac{1}{\e^{\upbeta(-p_1)\lVert\pvec\rVert} - 1} + \frac{1}{\e^{\upbeta(p_1)\lVert\pvec\rVert} - 1} \right) \, \Diff3\pvec = 0 \, ,
	\end{align*}
	and, analogously, $\Tcal_{32} = \Tcal_{23} = \left( \partial_2 \otimes \partial_3 \widetilde{W}_N  \right)(x,x) = 0$.
\end{proof}


\mysepline

\singlespacing
\renewcommand{\refname}{\Large References}
\begin{footnotesize}
	\providecommand{\etalchar}[1]{$^{#1}$}
	\providecommand{\doi}[1]{\url{https://doi.org/#1}}
	
\end{footnotesize}


\end{document}